\documentclass[10pt, oneside]{article}   	
\usepackage[margin=1in]{geometry} 
\usepackage{booktabs}
\usepackage{graphicx}
\usepackage{color}
\usepackage{hyperref}
\usepackage{amssymb}
\usepackage{amsmath}
\usepackage{braket}
\usepackage{mathtools}
\usepackage{comment}
\usepackage{algorithm, algorithmic, float}
\usepackage{authblk}
\usepackage{cleveref}
\usepackage{fancyhdr}
\fancypagestyle{plain}{
\fancyhead{}
}
\usepackage[utf8]{inputenc}

\usepackage{amsthm}

\hypersetup{colorlinks}

\newtheorem{theorem}{Theorem}

\newtheorem{lemma}[theorem]{Lemma}

\newcommand{\secref}[1]{Section~\ref{sec:#1}}

\newcommand{\thmref}[1]{Theorem~\ref{thm:#1}}

\DeclareMathOperator{\poly}{poly}
\DeclareMathOperator{\bbE}{\mathbb{E}}
\newcommand{\eps}{\epsilon}

\title{Adaptive Quantum Simulated Annealing for Bayesian Inference and Estimating Partition Functions}
\author[1]{Aram W. Harrow}
\author[1]{Annie Y. Wei}
\affil[1]{Center for Theoretical Physics, Massachusetts Institute of Technology}
\date{}

\begin{document}
\maketitle

\begin{abstract}
  Markov chain Monte Carlo algorithms have important applications in counting problems and in machine learning problems, settings that involve estimating quantities that are difficult to compute exactly.   How much can quantum computers speed up classical Markov chain algorithms? In this work we consider the problem of speeding up simulated annealing algorithms, where the stationary distributions of the Markov chains are Gibbs distributions at temperatures specified according to an annealing schedule.

  We construct a quantum algorithm that both adaptively constructs an annealing schedule and quantum samples at each temperature.  Our adaptive annealing schedule roughly matches the length of the best classical adaptive annealing schedules and improves on nonadaptive temperature schedules by roughly a quadratic factor.  Our dependence on the Markov chain gap matches other quantum algorithms and is quadratically better than what classical Markov chains achieve.  Our algorithm is the first to combine both of these quadratic improvements. Like other quantum walk algorithms, it also improves on classical algorithms by producing ``qsamples'' instead of classical samples. This means preparing quantum states whose amplitudes are the square roots of the target probability distribution.

In constructing the annealing schedule we make use of amplitude estimation, and we introduce a method for making amplitude estimation nondestructive at almost no additional cost, a result that may have independent interest. Finally we demonstrate how this quantum simulated annealing algorithm can be applied to the problems of estimating partition functions and Bayesian inference.
\end{abstract}


\section{Introduction}
Grover search yields a quadratic speedup over classical exhaustive search for the problem of unstructured search. A major challenge in quantum algorithms is to extend this quadratic speedup to more structured search problems.  One particularly important case is Markov chain Monte Carlo algorithms, which make it possible to efficiently sample from the stationary distribution of a Markov chain. Markov chain Monte Carlo methods have applications both in Bayesian inference, where such methods are used to sample from a posterior distribution which might otherwise be difficult to compute directly, and in counting problems~\cite{jsv, dfk} via the connection between approximate counting and sampling.

However, it is currently an open question whether there exists a completely quantum analog of the classical Markov Chain Monte Carlo algorithm. While quantum walks~\cite{szegedy} yield quadratically faster mixing in a variety of special cases~\cite{richter}, there is no general quadratic speedup known for MCMC sampling. Classical Markov chains are known to mix in time $O(\delta^{-1}\,\log(1/\min_x\Pi(x)))$~\cite{aldous}, where $\delta$ is the spectral gap of the Markov chain, and $\Pi(x)$ denotes the stationary distribution, while in the most general case quantum Markov chains have been shown to mix in time $O(1/\sqrt{\delta\, \min_x \Pi(x)})$~\cite{magniez}. Even though a recent result~\cite{amb} achieved a quadratic speedup in hitting time for the problem of searching for marked elements, the technique used there, that of quantum fast-forwarding~\cite{as}, will not yield a quadratic speedup for MCMC sampling as it also scales like $O(1/\min_x \Pi(x))$. In the regime where $(1/\min_x \Pi(x))$ scales with the size of the search space, the resulting quantum scaling is exponentially worse than the scaling of classical MCMC.
 
Indeed, there are well-known barriers to a general quantum speedup. First, directed Markov chains are general enough to encompass any randomized classical algorithm, but there are oracle problems, such as parity, for which quantum algorithms cannot obtain more than a constant speedup, so any such speedup would need to rely on structural features of the Markov chain. Second, many natural quantum walks that produce a classical sample do so by measuring a state whose amplitudes are all nonnegative reals, which means that they could prepare such a state at no extra cost.  Such a state is called a qsample~\cite{aharonov} and is the coherent encoding of the stationary distribution of the classical Markov chain. If qsamples could be prepared even polynomially more slowly than the mixing time of classical Markov chains, let alone quadratically faster, then this would imply the unlikely conclusion that $\text{SZK}\subseteq\text{BQP}$~\cite{aharonov, orsucci}.  

As a result, there are several distinct approaches to the problem of qsampling and state generation, and we briefly survey these approaches in \Cref{sec:related}. The approach that we shall employ, that of quantum simulated annealing (QSA)~\cite{somma2, somma, wocjan, yungaag}, relies on qsampling the stationary distributions of a series of intermediate Markov chains. Successive stationary distributions satisfy a ``slow-varying condition" $|\langle\Pi_i|\Pi_{i+1}\rangle|^2\geq\text{const}$, which allows these algorithms to bound the dependence on $\min_x \Pi(x)$ while preserving the $O(1/\sqrt{\delta})$ square root scaling in the spectral gap. Such algorithms do so at the cost of also scaling with the length of the annealing schedule $\ell$, and in this work we will show how to reduce the length $\ell$. 

Our work relies on two previous algorithmic results.  First is the QSA algorithm of Wocjan and Abeyesinghe~\cite{wocjan}, who showed how to qsample from the last of a series of Markov chains.  Specifically, given a series of $\ell$ Markov chains such that the first Markov chain is easy to qsample, all the spectral gaps are lower bounded by $\delta$, and the stationary states have constant overlap, qsampling from the last Markov chain can be performed using $\tilde{O}(\ell/\sqrt{\delta})$ total Markov chain steps.  This is important because quantum walks naturally yield reflections about the stationary state, so this gives an efficient way to turn the ability to reflect into the ability to qsample.  However, it does not give us a good way to bound the length $\ell$.  If $Z = \sum_x e^{-H(x)}$ for some $H(x)\geq 0$ then we can naively bound $\ell \leq \max_x H(x)$.  A somewhat better bound is $\ell \lessapprox F := \log(1/Z)$.
We use the notation $F$ because this quantity is called the ``free energy'' in statistical physics.
More precisely, $\ell \leq (1+F) \log \log |\Omega|$ where $\Omega$ is the state space, and this sequence can be found knowing only a bound on $F$; see Lemma 3.2 of \cite{stefankovic}.  This linear scaling with $F$ cannot be improved for such nonadaptive schedules.

However, a better sequence of Markov chains can be found if we are willing to choose them {\em adaptively}, i.e.~based on information we extract from our samples as we run the algorithm.   The second result we use is due to {\v{S}}tefankovi{\v{c}, Vempala, and Vigoda (SVV)~\cite{stefankovic}, who gave a classical algorithm for finding adaptive sequences of Markov chains of length $\tilde O(\sqrt{F})$, an almost quadratic improvement.  (Note that~\cite{huber} gives a simpler classical algorithm for finding quadratically shorter sequences, but it requires that the Hamiltonian not change sign, limiting its application beyond counting problems.) At first glance, such adaptive algorithms appear difficult to quantize since extracting information from qsamples, say in order to determine the adaptive sequence, will generally damage the states.  Indeed, the only quantum algorithm to use SVV was Montanaro's~\cite{montanaro} quantum algorithm for summing partition functions, which uses the QSA algorithm of Wocjan and Abeyesinghe~\cite{wocjan} to partially quantize a classical algorithm for summing partition functions.  However, while \cite{montanaro} could {\em use} the adaptive sequence in its quantum algorithm, it had to rely on classical methods to {\em compute} the sequence from SVV, which limited its quantum speedup.

Our work combines the QSA algorithm of Wocjan and Abeyesinghe~\cite{wocjan} with a fully quantized version of the work of SVV, achieving a runtime of $\tilde O(\sqrt{F/\delta})$.  In other words we adaptively obtain a sequence matching the length from SVV (i.e.~$\ell=\tilde O(\sqrt F)$) while also achieving the square-root scaling with $1/\delta$ from previous QSA algorithms~\cite{somma2, somma, wocjan}. In doing so we show that amplitude estimation~\cite{brassard} can be made nondestructive using a state restoration scheme inspired by~\cite{temme}, a result we believe will be useful in its own right. 

We also show that this algorithm can be applied both to the problem of estimating the partition function in counting problems and to the problem of Bayesian inference, as both problems share a general structure. In the counting problem we have a partition function of the form $Z(\beta)=\sum_{k=0}^n a_k e^{-\beta k}$, and we would like to estimate the quantity $Z(\infty)=a_0$, which is hard to compute, by annealing from $Z(0)$. In the Bayesian inference problem we have a prior $\Pi_0(\theta)$ and a likelihood function $L(\theta)$, and we would like to sample from the hard-to-compute posterior distribution $\Pi_1(\theta)=\Pi_0(\theta)L(\theta)/Z$ by annealing through the intermediate distributions $\Pi_\beta(\theta)=\Pi_0(\theta)\exp(\beta L(\theta))/Z_\beta$. We obtain the following theorem as our main result, which we also summarize in Table \ref{tab:summary}.

\begin{table*}[h]
\begin{tabular}{l @{$\quad$} l @{$\quad$} l @{$\quad$} l }
  \toprule
  \textbf{Problem} & \textbf{Our Result} & \textbf{Best Previous Result} & \textbf{Best Classical Result}\\
\midrule
Counting Problems & $\tilde{O}(\log |\Omega|/(\sqrt{\delta}\epsilon))$ & $\tilde{O}(\log |\Omega|/(\sqrt{\delta}\epsilon)+\log |\Omega|/\delta)$ & $\tilde{O}(\log |\Omega|/(\delta\epsilon^2))$\\
Bayesian Inference & $\tilde{O}(\sqrt{\bbE_{\Pi_0}[L(\theta)]/\delta})$ & $\tilde{O}(\max_\theta L(\theta)/\sqrt{\delta})$ & $O(\max_\theta L(\theta)/\delta)$\\
  \bottomrule
\end{tabular}
\caption{\label{tab:summary} Summary of main results. Here $\delta$ denotes the spectral gap of the Markov chain. Letting $n$ be the maximum upper range for the counting problem (equivalently, the maximum value of the Hamiltonian), typically $|\Omega|=Z(0)\sim\exp(n)$ and $\delta\sim\poly(n)$. $L(\theta)$ denotes the likelihood function for the Bayesian inference problem and likewise corresponds to values of the Hamiltonian.  Our results are formalized in \Cref{thm:qaaa,thm:counting}.  The previous best [quantum] results for counting and Bayesian inference are due to Montanaro~\cite{montanaro} and Wocjan-Abeyesinghe~\cite{wocjan} respectively.  The classical algorithm for counting is due to {\v{S}}tefankovi{\v{c}}, Vempala and Vigoda~\cite{stefankovic} and the algorithm for Bayesian inference simply uses simulated annealing with the nonadaptive schedule in \cite{stefankovic}.}
\end{table*}

\begin{theorem}[Informal statement of main results]\label{thm:main}~
  \begin{enumerate}
  \item {\em Bayesian inference.} Given a prior $\Pi_0(\theta)$ and a likelihood function $L(\theta)$, define distributions  $\Pi_\beta(\theta)\propto \Pi_0(\theta)\exp(\beta L(\theta))$  for $\beta\in [0,1]$.  Suppose that for each $\beta$ we can compute a Markov chain $M_\beta$ with stationary distribution $\Pi_\beta$ and with gap $\geq \delta$.  Then we can qsample from $\ket{\Pi_1}$ using $\tilde O(\sqrt{\bbE_{\Pi_0}[L(\theta)]/\delta})$ steps of the quantum walks corresponding to various $M_\beta$.
    \item {\em Estimating partition functions.} Let $Z(\beta)=\sum_x  e^{-\beta H(x)}$ with $H(x)\geq 0$ and suppose again that we have access to Markov chains $M_\beta$ with gaps $\geq \delta$ and stationary distributions $\propto e^{-\beta H(x)}$. Then we can estimate $Z(\infty)$ to multiplicative error $\eps$ with high probability using $\tilde O(\log(Z(0))/\sqrt{\delta}\eps)$ steps of the quantum walks corresponding to $M_\beta$.
\end{enumerate}
\end{theorem}

These are formalized as Theorems \ref{thm:qaaa} and \ref{thm:counting}.  In each case we match the schedule length of SVV's adaptive algorithm and the gap dependence of Wocjan-Abeyesinghe, thus improving on all previous algorithms. An important subroutine in our results is a nondestructive version of amplitude estimation, formalized below in \thmref{amplest1} and described in detail in Section \ref{sec:naa}.

We also consider applications of the partition function algorithm to representative problems from statistical physics and computer science, again improving on previous algorithms. Our results are summarized in Table \ref{tab:examples} and discussed in more detail in \Cref{subsec:partitionfns}.

This paper is organized as follows: in the rest of this introduction we briefly survey related work and provide a technical overview of our work. In \Cref{sec:existencesched} we show that there exists an adaptive cooling schedule by slightly modifying the arguments of SVV to also work in the Bayesian inference case. This adaptive cooling schedule then translates into a temperature schedule that is quadratically shorter than any nonadaptive schedule in both the Bayesian inference and counting problem cases. In \Cref{sec:find-cool} we describe the quantum algorithm which both constructs the adaptive cooling schedule and anneals to the quantum sample at each temperature. Applying this algorithm to Bayesian inference and the counting problem, we establish our main result \Cref{thm:main}, formalized as \Cref{thm:qaaa,thm:counting}. In \Cref{sec:naa} we describe, in detail, how to perform state restoration following amplitude estimation at almost no additional cost. In \Cref{subsec:partitionfns} we consider applications of the partition function algorithm to representative problems from statistical physics and computer science, and in \Cref{subsec:warm-starts} we discuss warm starts for speeding up Markov chain mixing times, as well as how they have been incorporated into the algorithms of \Cref{sec:find-cool}.  Our conclusion is in \Cref{sec:conclusion}.

\begin{table*}[h]
\begin{tabular}{l @{$\quad$} l @{$\quad$} l @{$\quad$} l }
\hline
\hline
\textbf{Problem} & \textbf{Our Result} & \textbf{Best Previous Result} & \textbf{Best Classical Result}\\
\hline
Counting $k$-colorings & $\tilde{O}(|V|^{3/2}/\epsilon)$ &  $\tilde{O}(|V|^{3/2}/\epsilon+|V|^2)$ & $\tilde{O}(|V|^2/\epsilon^2)$\\
Ising model & $\tilde{O}(|V|^{3/2}/\epsilon)$ & $\tilde{O}(|V|^{3/2}/\epsilon+|V|^2)$ & $\tilde{O}(|V|^2/\epsilon^2)$\\
Counting matchings & $\tilde{O}(|V|^{3/2}|E|^{1/2}/\epsilon)$ & $\tilde{O}(|V|^{3/2}|E|^{1/2}/\epsilon+|V|^2|E|)$ & $\tilde{O}(|V|^2|E|/\epsilon^2)$\\
Counting independent sets & $\tilde{O}(|V|^{3/2}/\epsilon)$ & $\tilde{O}(|V|^{3/2}/\epsilon+|V|^2)$ & $\tilde{O}(|V|^2/\epsilon^2)$\\
\hline
\hline
\end{tabular}
\caption{\label{tab:examples} Summary of applications to estimating the partition function in counting problems.  See the text of \cref{subsec:partitionfns} for discussion and references.}
\end{table*}

\subsection{Related Work}\label{sec:related}
Here we briefly describe alternative approaches to the problem of qsampling and state generation, noting some benefits and drawbacks of each approach when compared with QSA. 

\begin{itemize}
\item \textbf{Direct generation}: An approach due to Zalka~\cite{zalka}, rediscovered independently by Grover and Rudolph~\cite{groverrudolph} and Kaye and Mosca~\cite{kayemosca}, generates the state directly via rotations, but its scope is limited as it is only efficient in the special case where the probability distribution is efficiently integrable. 

\item \textbf{Adiabatic state generation}: Aharanov and Ta-Shma~\cite{aharonov} offer an approach to qsampling via adiabatic computing, but their approach scales like $O(1/\delta)$ in the spectral gap. Thus, while it produces qsamples instead of samples, it offers no speedup over the classical case.

\item \textbf{Metropolis sampling}: An approach by~\cite{temme} that relies on Metropolis sampling generalizes qsampling to quantum Hamiltonians, but it likewise scales like $O(1/\delta)$ in the spectral gap.~\cite{yungaag} combines Metropolis sampling with QSA to extend the $O(\ell/\sqrt{\delta})$ scaling of QSA to quantum Hamiltonians, but the scaling is otherwise equivalent to that of other QSA algorithms.

\item \textbf{Quantum rejection sampling}: In quantum rejection sampling~\cite{rejectionsampling, low, wiebe}, to obtain target state $\ket{\Pi}$ we instead prepare some superposition of the desired state $\ket{\Pi}$ and an undesired state $\ket{\Pi^\perp}$ and then apply amplitude amplification to obtain $\ket{\Pi}$. As~\cite{wiebe} notes, this scheme is generally inefficient;  to deal with this,~\cite{low} specializes to the case of distributions structured as a Bayesian network, while~\cite{wiebe} employs semi-classical Bayesian updating. Even then, the algorithm of~\cite{wiebe} still scales like $O(1/\sqrt{\epsilon Z})$ per update in $\epsilon$, the approximation error, and $Z$, the partition function of the posterior distribution, whereas our algorithm's scaling is $\sim \sqrt{\delta^{-1}\log(1/Z)}\log(1/\epsilon)$. (These scalings depend on the normalization convention used for $Z$; see \secref{adaptive-intro}.)  This scaling is generally better because $\delta$ can often be improved with a good choice of Markov chain, and when these chains are rapidly mixing $1/\delta$ will be $\poly\log(1/Z)$.
\end{itemize}

\subsection{Technical Overview}
Here we describe adaptive annealing schedules and their application to counting problems and Bayesian inference. Then we describe our quantum algorithm for finding and annealing through such a schedule.

\subsubsection{Adaptive Annealing Schedules for Counting Problems and Bayesian Inference}\label{sec:adaptive-intro}

In both the counting problem and the Bayesian inference problem, we have a partition function of the form
\begin{equation}\label{eq:z}
Z(\beta)=\sum_{x\in\Omega} e^{-\beta H(x)}
\end{equation}
at inverse temperature $\beta$, with a Hamiltonian we denote by $H(x)$ for some random variable $x$ over state space $\Omega$.  We assume that $H(x)$ is easy to compute (say a sum of local terms) and $\Omega$ is also a simple set, such as $\{0,1\}^n$, although it may also be a non-product set such as the set of permutations.
Such a partition function corresponds to the normalization of the Gibbs distribution at inverse temperature $\beta$, which is given by
\begin{equation}
\Pi_\beta(x)=\frac{e^{-\beta H(x)}}{Z(\beta)}.
\end{equation}

In the counting problem of SVV~\cite{stefankovic} and Montanaro~\cite{montanaro}, the Hamiltonian takes on values $k\in\{0,\ldots,n\}$ corresponding to a discrete quantity we would like to count, such as the number of colorings of a graph, or the number of matchings. In \Cref{subsec:partitionfns} we give several examples of problems from statistical physics and computer science that can be framed in this form. In such problems we would have a partition function of the form
\begin{equation}\label{eq:zcount}
Z(\beta)=\sum_{k=0}^na_ke^{-\beta k},
\end{equation}
where $a_k = |H^{-1}(k)|$.  In general we do not need the energy function to take on only integer values but it will be convenient to assume that $0\leq H(x)\leq  n$ for all $x$. 

We want to estimate the quantity $Z(\infty)=a_0$, which is often difficult to compute, while $Z(0)=\sum_ka_k = |\Omega|$, corresponding simply to the size of the parameter space, is easy to compute. The idea is to establish a schedule of $\ell+1$ inverse temperatures $\beta_0, \beta_1,\ldots,\beta_\ell$, with $\beta_0=0$ and $\beta_\ell=\infty$, known as a cooling schedule, that allows us to anneal from the easy case of $\beta=0$ to the hard case of $\beta=\infty$. Once we have a cooling schedule, we can sample from the Gibbs distribution at each inverse temperature $\beta_i$, given by
\begin{equation}
\Pi_{\beta_i}(x)=\frac{e^{-\beta_iH(x)}}{Z(\beta_i)}.
\end{equation}
Then, for $x$ sampled from $\Pi_{\beta_i}$, the quantity
\begin{equation}
W_{\beta_i,\beta_{i+1}}(x)=e^{(\beta_i-\beta_{i+1})H(x)}
\end{equation}
has expectation value
\begin{equation}\label{eq:w}
\bbE_{\Pi_{\beta_i}}[W_{\beta_i,\beta_{i+1}}]=\frac{Z(\beta_{i+1})}{Z(\beta_i)},
\end{equation}
so we can calculate $Z(\infty)$ as the telescoping product
\begin{equation}\label{eq:telescope}
Z(\infty)=Z(0)\frac{Z(\beta_1)}{Z(0)}\frac{Z(\beta_2)}{Z(\beta_1)}\ldots\frac{Z(\infty)}{Z(\beta_{\ell-1})}
\end{equation}
by sampling $W_{\beta_i, \beta_{i+1}}$ at each successive temperature. In the SVV algorithm the temperature schedule is determined adaptively using properties of $\log Z(\beta)$ like convexity, so that, letting $|\Omega|=Z(0)$, the schedule has length $\ell=O(\sqrt{\log |\Omega|}\log n\log\log |\Omega|)=\tilde{O}(\sqrt{\log |\Omega|})$, a quadratic improvement over the best possible non-adaptive schedule length of $O(\log |\Omega|\log n)=\tilde{O}(\log |\Omega|)$. Recall that $n=\max_x H(x)$.  We write $\tilde O(f)$ to suppress terms that are polylog in $f$, and in doing so, we assume that $\log|\Omega|$ and $n$ are polynomially related.  Our results do not otherwise assume any relation between $|\Omega|$ and $n$.

Such techniques could also be applied to the problem of Bayesian inference. Bayesian inference refers to an important paradigm in machine learning where values are assigned to model parameters according to a probability distribution that is updated using the observed data; this then allows us to quantify our uncertainty in the model parameters, as well as to update this uncertainty. Given model parameters $\theta$ that we wish to learn, we generally start with a prior distribution $\Pi_0(\theta)$ over the possible values that $\theta$ can take, and then given data points $\{x_i\}$ we update our prior distribution to obtain a posterior distribution over $\theta$ according to Bayes' rule:
\begin{equation}
p(\theta|\{x_i\})=\frac{\Pi_0(\theta)\prod_ip(x_i|\theta)}{\sum_\theta \Pi_0(\theta)\prod_ip(x_i|\theta)}.
\end{equation}
Here the normalization, the partition function $Z=\sum_\theta \Pi_0(\theta)\prod_ip(x_i|\theta)$, is often difficult to compute directly due to the sheer size of the parameter space. In analogy to the counting problem, where $\beta$ parametrizes the partition function from the easy case of $Z(0)$ to the hard case of $Z(\infty)$, for Bayesian inference we will define the partition function
\begin{equation}\label{eq:zlambda}
Z(\beta)=\sum_\theta\Pi_0(\theta)e^{-\beta L(\theta)},
\end{equation}
where the Hamiltonian corresponds to $L(\theta)$, the negative log-likelihood function, defined as
\begin{equation}
L(\theta)=-\log\left(\prod_ip(x_i|\theta)\right).
\end{equation}
Then, in analogy to the counting problem, $Z(0)$ is easy to calculate as it just corresponds to $\sum_\theta\Pi_0(\theta)=1$, while $Z(1)$, corresponding to the full posterior distribution, is hard to compute. As in the counting problem, we can imagine establishing a temperature schedule $\beta_0, \beta_1,\ldots,\beta_\ell$ with $\beta_0=0$ and $\beta_\ell=1$. Then the Gibbs distribution at each temperature is given by
\begin{equation}
\Pi_{\beta_i}(\theta)=\frac{\Pi_0(\theta)e^{-\beta_i L(\theta)}}{Z(\beta_i)}.
\end{equation}
Note, however, that in the case of Bayesian inference we don't need to compute the actual value of the partition function $Z(1)$ since we're ultimately interested in sampling from the posterior distribution. That is, it's enough to just return a sample from the last Markov chain. Thus we can in fact think of our Bayesian inference algorithm as performing simulated annealing using the adaptive cooling schedule as an annealing schedule. Because of the similarities between the counting problem and the Bayesian inference problem, we claim that we can modify the arguments of SVV to show that there exists a temperature schedule for Bayesian inference of length $\ell=\tilde{O}(\sqrt{\log(1/Z(1))})=\tilde{O}(\sqrt{\bbE_{\Pi_0}[L(\theta)]})$. This schedule is quadratically shorter than the nonadaptive annealing schedule obtained in QSA papers such as those of~\cite{somma2, somma, wocjan}, where the best result due to~\cite{wocjan} uses inverse temperatures separated by a constant $\Delta\beta=O(1/\|H\|)$ so that $\ell=O(\|H\|)=O(\max_\theta L(\theta))$. Additionally, the dependence on $(1/Z(1))$ is exponentially better than the $O(1/\sqrt{Z(1)})$ dependence per Bayesian update in algorithms based on quantum rejection sampling, like that of~\cite{wiebe}; however, this advantage is partially offset by a new dependence on the gap $\delta$.

Merely the existence of such a short temperature schedule is not quite enough.  In the next section we will demonstrate a quantum algorithm for efficiently finding temperature schedules of this length.

\subsubsection{Quantizing Adaptive Annealing}

So far our claims, that Bayesian inference can be treated as a simulated annealing problem analogous to counting problems, and that the annealing schedule can be made quadratically shorter, have been claims that would apply equally to both classical and quantum settings. We additionally claim that the computation of cooling schedules can be fully quantized. Combined with the QSA algorithm of~\cite{wocjan}, which performs quantum annealing given a cooling schedule, this means that it is possible to fully quantize both the algorithm for computing partition functions in the counting problem, and the algorithm for qsampling from the posterior distribution in Bayesian inference. To prove this claim, we combine techniques from~\cite{montanaro} with a nondestructive version of amplitude estimation.

Montanaro's~\cite{montanaro} algorithm for summing partition functions partially quantizes the SVV algorithm; the adaptive temperature schedule itself is still computed classically according to the SVV algorithm, but once given an inverse temperature, the algorithm specifies how to quantum sample at that temperature, as well as how to efficiently compute expectation values using those samples. Qsampling is performed according to the QSA algorithm of Wocjan and Abeyesinghe~\cite{wocjan}, who showed that given a sequence of $\ell$ slow-varying Markov chains (i.e., the overlap between successive stationary distributions is lower-bounded by some constant), each with spectral gap at least $\delta$, an approximation to the stationary distribution of the final Markov chain can be obtained with $\tilde{O}(\ell/\sqrt{\delta})$ Markov chain steps, whereas classically the dependence on $\ell$ and $\delta$ would be $O(\ell/\delta)$. Given these quantum samples, Montanaro's algorithm then estimates expectation values using an amplitude estimation based algorithm that requires quadratically fewer samples than would be necessary classically. Overall Montanaro shows that it is possible to estimate the partition function with up to $\epsilon$ multiplicative error using $\tilde{O}(\log |\Omega|/(\sqrt{\delta}\epsilon)+\log |\Omega|/\delta)$ Markov chain steps, and notes that this complexity could be improved to $\tilde{O}(\log |\Omega|/(\sqrt{\delta}\epsilon))$ were it were possible to compute the cooling schedule itself via quantum means. The fact that the SVV algorithm uses a nonadaptive temperature schedule as a ``warm start" for the adaptive schedule (which allows for a faster mixing time) is cited as an obstacle to quantizing the computation of the cooling schedule. We claim that these obstacles can be overcome.

As in Montanaro's algorithm, we can use the algorithm of Wocjan and Abeyesinghe to sample from the Gibbs distribution at each temperature. Additionally, we will also quantize the actual process of computing the cooling schedule itself. Our algorithm works as follows: since we are guaranteed the existence of the adaptive cooling schedule, we can binary search to find the next temperature. For each binary search candidate we can use amplitude estimation to calculate the overlap between the candidate state and the current state, which allows us to check whether the slow-varying condition is satisfied. Note that amplitude estimation only requires that we be able to reflect over the candidate state, and that the quantum walk operator provides such a reflection operator. We also observe that all quantum measurements occur only during the amplitude estimation step, and that amplitude estimation can be made non-destructive so that it's possible to restore the post-measurement state to the pre-measurement state at almost no additional cost.  
Finally we also claim that in the quantum case, the slow-varying condition itself is enough to ensure warm-start mixing times, which ends up simplifying one of the steps in the SVV algorithm.

Putting these claims together yields the results of Theorem \ref{thm:main} and Table \ref{tab:summary}.

\section{Existence of Cooling Schedule}\label{sec:existencesched}
In this section we slightly modify an argument of SVV~\cite{stefankovic} to show that there exists a cooling schedule of bounded length for a partition function of the form (\ref{eq:z}), which encompasses both (\ref{eq:zcount}), corresponding to the counting problem, and (\ref{eq:zlambda}), corresponding to Bayesian inference. Furthermore, this cooling schedule satisfies the Chebyshev condition in the case of counting problems, and the slow-varying condition in the case of Bayesian inference. 

As noted in the previous section, the SVV algorithm generates a sequence of inverse temperatures $\beta_0, \beta_1,\ldots,\beta_\ell$ with $\beta_0=0$ and $\beta_\ell=\infty$; then, given such a schedule, the idea is to sample from the Gibbs distribution at each temperature in order to compute the quantities $W_{\beta_i, \beta_{i+1}}$, whose expectation value is the ratio of $Z$ at successive temperatures. Taking the telescoping product of these ratios according to equation 
(\ref{eq:telescope}) then allows us to estimate $Z(\infty)$ starting from $Z(0)$.

In a $B$-Chebyshev cooling schedule such as that generated by SVV, we have the additional requirement that the variance of $W_{\beta_i, \beta_{i+1}}$ is bounded; that is, that
\begin{equation}\label{eq:boundedvariance}
\frac{\bbE\left(W^2_{\beta_i, \beta_{i+1}}\right)}{\bbE\left(W_{\beta_i, \beta_{i+1}}\right)^2}=\frac{Z(2\beta_{i+1}-\beta_i)Z(\beta_i)}{Z(\beta_{i+1})^2}\leq B
\end{equation}
for a constant $B$. This additional bounded variance requirement then guarantees that the product of expectation values $\bbE[W_{\beta_0, \beta_1}]\bbE[W_{\beta_1, \beta_2}]\ldots\bbE[W_{\beta_{\ell-1}, \beta_{\ell}}]$ will be a good approximation to the product $W_{\beta_0, \beta_1}W_{\beta_1, \beta_2}\cdots W_{\beta_{\ell-1}, \beta_{\ell}}$ within a bounded number of samples. 

In the case of Bayesian inference we're not actually trying to calculate the partition function $Z(1)$ (instead we want to sample from the Gibbs distribution at $\beta=1$), so it might seem like we don't need the additional bounded variance condition.
However, the {\em slow-varying condition} is another property, closely related to bounded variance, which we will need.  The slow-varying condition states that $|\langle\pi_{\beta_i}|\pi_{\beta_{i+1}}\rangle|^2\geq 1/B$, since
\begin{align*}
\langle\pi_{\beta_i}|\pi_{\beta_{i+1}}\rangle&=\frac{Z\left(\frac{\beta_i+\beta_{i+1}}{2}\right)}{\sqrt{Z(\beta_i)}\sqrt{Z(\beta_{i+1})}}.
\end{align*}
Then the slow-varying condition can be rewritten as
\begin{equation}\label{eq:boundedslovar}
\frac{Z(\beta_i)Z(\beta_{i+1})}{Z\left(\frac{\beta_i+\beta_{i+1}}{2}\right)^2}\leq B.
\end{equation}

We define $f(\beta)=\log Z(\beta)$ to help understand the slow-varying and Chebyshev conditions.  Note that $f$ is convex.  Observe that when we set $B=e^2$,  both the slow-varying condition (\ref{eq:boundedslovar}) and the Chebyshev condition (\ref{eq:boundedvariance}) can be rewritten in the form
\begin{equation}\label{eq:approxconcav}
f\left(\frac{\gamma_i+\gamma_{i+1}}{2}\right)\geq\frac{f(\gamma_i)+f(\gamma_{i+1})}{2}-1,
\end{equation}
where for the Chebyshev condition $\gamma_i=\beta_i$ and $\gamma_{i+1}=2\beta_{i+1}-\beta_i$, while for the slow-varying condition $\gamma_i=\beta_i$ and $\gamma_{i+1}=\beta_{i+1}$.  \Cref{eq:approxconcav} should be compared with the inequality $f(\frac{\gamma_i +\gamma_{i+1}}{2}) \leq
\frac{f(\gamma_i) +f(\gamma_{i+1})}{2}$ resulting from convexity of $f$.  

The existence of Chebyshev and slow-varying sequences is then expressed by the following lemma, which guarantees the existence of a sequence of inverse temperatures satisfying equation (\ref{eq:approxconcav}).  We will slightly modify the original bound that appears in SVV, from $\ell\leq\sqrt{(f(0)-f(1))\log(f'(0)/f'(\gamma))}$ to $\ell\leq\sqrt{(f(0)-f(1))\log(f'(0)/(f'(\gamma)+1))}$, in order for this bound to work in the case of Bayesian inference. The full proof of the lemma appears in Appendix \ref{app:schedlen}.
\begin{lemma}\label{lemma:svv}
(Modified from SVV~\cite{stefankovic} Lemma 4.3, Appendix \ref{app:schedlen}) For $f$ a convex function over domain $[0, \gamma]$, there exists a sequence $\gamma_0<\gamma_1<\ldots<\gamma_\ell$ with $\gamma_0=0$ and $\gamma_\ell=\gamma$ satisfying
\begin{equation}
f\left(\frac{\gamma_i+\gamma_{i+1}}{2}\right)\geq\frac{f(\gamma_i)+f(\gamma_{i+1})}{2}-1
\end{equation}
with length
\begin{equation}\label{eq:len}
\ell\leq \sqrt{(f(0)-f(\gamma))\log\left(\frac{f'(0)}{f'(\gamma)+1}\right)}.
\end{equation}
\end{lemma}

This suggests that we can construct a Chebyshev cooling schedule greedily; given left endpoint $\gamma_i$, choose the next endpoint by finding the largest possible right endpoint $\gamma_{i+1}$ so that the midpoint satisfies equation (\ref{eq:approxconcav}), and Lemma \ref{lemma:svv} then guarantees an upper bound on the length of a schedule constructed in this manner.

In the next section we will describe a quantum algorithm for efficiently carrying out a version of this procedure. In the remainder of this section we show that the length of the schedule generated by this algorithm, both in the case of the counting problem and in the case of Bayesian inference, is quadratically shorter than the length of the corresponding nonadaptive schedule.

In the case of the counting problem, SVV show that the schedule derived from (\ref{eq:len}) ends up being $\tilde{O}(\sqrt{\log |\Omega|})$, where typically $\log|\Omega|\sim\poly(n)$ for $n=\max_x H(x)$:
\begin{theorem}\label{thm:countinglen}
(SVV~\cite{stefankovic} Theorem 4.1) For $Z(\beta)$ a partition function of the form given by equation (\ref{eq:zcount}), letting $|\Omega|=Z(0)$ and assuming $Z(\infty)\geq 1$, there exists a $B$-Chebyshev cooling schedule with $B=e^2$, $\beta_0=0$, and $\beta_\ell=\infty$, of length
\[
O(\log\log |\Omega|\sqrt{\log|\Omega|\log(n)})=\tilde{O}(\sqrt{\log |\Omega|}).
\]
\end{theorem}
The full proof of Theorem (\ref{thm:countinglen}) can be found in~\cite{stefankovic}, but the idea is the following.   For counting problems, where we need to anneal all the way to $\beta_\ell=\infty$, it's enough to take $\beta_{\ell-1}=\gamma$ with $\gamma$ the inverse temperature satisfying $f(\gamma)=1$.  This choice of $\gamma$ guarantees that \cref{eq:approxconcav} is satisfied between $\beta_\ell=\infty$ and $\beta_{\ell-1}=\gamma$. Next we use Lemma (\ref{lemma:svv}) to note that there exists a sequence of $\gamma_{0'},\gamma_{1'},\ldots,\gamma_{\ell'}$ with $\gamma_{0'}=\gamma_0=0$ and $\gamma_{\ell'}=\gamma$ that satisfy (\ref{eq:approxconcav}) with
\begin{equation}
\ell'\leq\sqrt{\log|\Omega|\log(n)}.
\end{equation}
We can see that the expression for $\ell'$ comes from (\ref{eq:len}) with $\log |\Omega|$ corresponding to the $f(0)-f(\gamma)$ term and $\log n$ corresponding to the $\log(f'(0) / (f'(\gamma) + 1))$ term. Next we need to extract the $\beta_0,\ldots,\beta_{\ell-1}$ from the $\gamma_{0'},\ldots,\gamma_{\ell'}$, where $\beta_0=\gamma_0=0$ and $\beta_{\ell-1}=\gamma_{\ell'}=\gamma$. SVV show that it suffices to insert additional inverse temperatures in each interval $[\gamma_i, \gamma_{i+1}]$ in the following way:
\[
\gamma_i, \gamma_i + (1/2)(\gamma_{i+1}-\gamma_i), \gamma_i + (3/4)(\gamma_{i+1}-\gamma_i), \gamma_i + (7/8)(\gamma_{i+1}-\gamma_i),\ldots, \gamma_i+(1-2^{-\lceil\log\log |\Omega|\rceil})(\gamma_{i+1}-\gamma_i), \gamma_{i+1},
\]
which ensures that each pair of adjacent temperatures satisfies the Chebyshev condition (\ref{eq:boundedvariance}). This adds an additional factor of $\log\log |\Omega|$, so the dominant term is still $\sqrt{\log |\Omega|}$.  SVV also show that any nonadaptive schedule must be $\tilde{\Omega}(\log |\Omega|)$, so the adaptive schedule is quadratically shorter.

In the case of Bayesian inference, we claim that we have a similar result, where the adaptive schedule has length $\ell=\tilde{O}(\sqrt{\log(1/Z(1))})=\tilde{O}(\sqrt{\bbE_{\Pi_0}[L(\theta)]})$. Here the argument is more straightforward because we can directly take the $\gamma_0,\ldots,\gamma_\ell$ from Lemma (\ref{lemma:svv}) to be the inverse temperatures $\beta_0,\ldots,\beta_\ell$.

\begin{theorem}\label{thm:len}
For partition function $Z(\beta)$ of the form given by equation (\ref{eq:zlambda}), there exists a temperature schedule with $B=e^2$, $\beta_0=0$, and $\beta_\ell=1$, satisfying $|\langle\Pi_{\beta_i}|\Pi_{\beta_{i+1}}\rangle|^2\geq 1/B$, of length
\begin{equation}
O\left(\sqrt{\bbE_{\Pi_0}[L(\theta)]\log(\bbE_{\Pi_0}[L(\theta)])}\right)=\tilde{O}\left(\sqrt{\bbE_{\Pi_0}[L(\theta)]}\right).
\end{equation}
\end{theorem}
\begin{proof}[Proof of Theorem \ref{thm:len}]
We use the result of Lemma \ref{lemma:svv} with $\gamma_\ell=1$. Plugging into the expression for the length of the cooling schedule from equation (\ref{eq:len}), we note that $f(0)=0$ and $f(1)=\log Z(1)$ so that $f(0)-f(1) = \log(1/Z(1))=-\log Z(1)$. Note that this can be rewritten as
\begin{align}
-\log Z(1)=&-\log\left(\sum_\theta\Pi_0(\theta)e^{-L(\theta)}\right)\nonumber\\
&=-\log\left(\bbE_{\Pi_0(\theta)}\left[e^{-L(\theta)}\right]\right).
\end{align}
By Jensen's inequality, $-\log(\bbE[X])\leq -\bbE[\log(X)]$, so
\begin{equation}
\log(1/Z(1))\leq \bbE_{\Pi_0}[L(\theta)].
\end{equation}
We also note that $f'(0)=\bbE_{\Pi_0(\theta)}[L(\theta)]$ and $f'(1)=\bbE_{\Pi_1(\theta)}[L(\theta)]$, where $\Pi_0(\theta)$ denotes the prior distribution and $\Pi_1(\theta)$ denotes the posterior distribution, so that $\log(f'(0) / (f'(1) + 1))\leq\log(\bbE_{\Pi_0}[L(\theta)])$. Putting everything together,
\begin{align}
\ell&=O\left(\sqrt{\bbE_{\Pi_0}[L(\theta)]\log(\bbE_{\Pi_0}[L(\theta)])}\right)\nonumber\\
&=\tilde{O}\left(\sqrt{\bbE_{\Pi_0(\theta)}[L(\theta)]}\right).
\end{align}
\end{proof}
The length of the adaptive cooling schedule is quadratically shorter than the length of nonadaptive annealing schedules currently employed by QSA algorithms such as those~\cite{somma2, somma, wocjan}. For example, in the best result due to~\cite{wocjan}, which employs slow-varying Markov chains to perform QSA on a sequence of Markov chains with stationary distributions given by
\begin{equation}
\Pi_\beta(x)=\frac{e^{-\beta H(x)}}{Z(\beta)},
\end{equation}
taking the inverse temperatures to be separated by a constant $\Delta\beta=1/\|H\|$ ensures that the slow-varying condition is preserved. Applying this to Bayesian inference, where we have $x=\theta$ with $\theta\sim \Pi_0(\theta)$, $H(\theta)=L(\theta)$, and $\beta$ that we anneal between 0 and 1, we end up with a nonadaptive schedule of length $O(\max_\theta L(\theta))$.

\section{Construction of Cooling Schedule and Quantum Algorithm Details}\label{sec:find-cool}
We now give a quantum algorithm that adaptively constructs the cooling schedule from the previous section. As it does so, it simultaneously produces the quantum state corresponding to the Gibbs distribution at the current inverse temperature in the schedule construction process. In the case of Bayesian inference, obtaining the state at the final inverse temperature corresponds to qsampling from the posterior distribution. For the counting problem, sampling at each inverse temperature allows us to estimate the telescoping product (see equation (\ref{eq:telescope})) corresponding to the partition function $Z(\infty)$.

To do so we will need the following result of Wocjan and Abeyesinghe~\cite{wocjan}, as restated by Montanaro~\cite{montanaro}, which shows that it is possible to quantum sample given access to a sequence of slow-varying Markov chains:
\begin{theorem}\label{thm:qsa}
(Wocjan and Abeyesinghe~\cite{wocjan}, restated as Montanaro~\cite{montanaro} Theorem 9) Assume that we have classical Markov chains $M_0,\ldots,M_\ell$ with stationary distributions $\Pi_0,\ldots,\Pi_\ell$ that are slow-varying; that is to say, they satisfy $|\langle \Pi_i|\Pi_{i+1}\rangle|^2\geq p$ for all $i=0,..,\ell-1$. Let $\delta$ lower bound the spectral gaps of the Markov chains, and assume that we can prepare the starting state $\ket{\Pi_0}$. Then, for any $\epsilon>0$, there is a quantum algorithm that produces a quantum state that is $\epsilon$-close to $\ket{\Pi_\ell}$ and uses
\[
O\left(\ell\sqrt{\delta^{-1}}\log^2(\ell/\epsilon)(1/p)\log(1/p)\right)
\]
total steps of the quantum walk operators $W_i$ corresponding to the Markov chains $M_i$.
\end{theorem}
As we stated in the previous section, satisfying the slow-varying condition takes the same form as satisfying the Chebyshev condition for a cooling schedule. We can also easily prepare the starting state $\ket{\Pi_0}$ using a result of~\cite{zalka, groverrudolph, kayemosca}, who showed that it is possible to efficiently create the coherent encoding
\[
  \sum_i\sqrt{p_i}\ket{i}
\]
of the discretized version $\{p_i\}$ of a probability distribution $p(x)$, provided that $p(x)$ can be efficiently integrated classically (for example, by Monte Carlo methods). For counting problems, $\ket{\Pi_0}$ is just the uniform distribution, which can be easily integrated. For Bayesian inference we make a choice of prior that can be integrated classically, allowing us to easily prepare $\ket{\Pi_0}$.  A recent review of other state-preparation methods can be found in \cite{Tang19}; see also \cite{aharonov}.

Now we describe how to proceed to the next state $\ket{\Pi_{\beta_{i+1}}}$ assuming that we already have the state $\ket{\Pi_{\beta_i}}$. According to the procedure described in the previous section, we'd like to find the largest $\beta_{i+1}$ so that $|\langle \Pi_{\beta_{i+1}}|\Pi_{\beta_i}\rangle|^2\geq p$ in the case of Bayesian inference, and $|\langle \Pi_{2\beta_{i+1}-\beta_i}|\Pi_{\beta_i}\rangle|^2\geq p$ for counting problems. To do so we will binary search for $\beta_{i+1}$ in the Bayesian case, and $2\beta_{i+1}-\beta_i$ in the counting problem case, computing the overlap for the state $\ket{\Pi_{\beta'}}$ corresponding to each candidate inverse temperature $\beta'$ to see if $|\langle \Pi_{\beta'}|\Pi_{\beta_i}\rangle|^2\geq p$ is satisfied. Note that we can't actually produce each state $\ket{\Pi_{\beta'}}$ since being able to anneal to this state would require that it already satisfy the slow-varying condition. Luckily, being able to reflect about $\ket{\Pi_{\beta'}}$ suffices, and quantum walks will give us the ability to perform this reflection. When we estimate the overlap we also need to make sure that the state $\ket{\Pi_{\beta_i}}$ is not destroyed, and we ensure this by computing the overlap between $\ket{\Pi_{\beta_i}}$ and $\ket{\Pi_{\beta'}}$ using a form of amplitude estimation that has been made nondestructive. This will be doubly useful in the case of counting problems, where to calculate $Z(\infty)$ we will need to estimate expectation values $\bbE[W_{\beta_i, \beta_{i+1}}]$ at intermediate temperatures without destroying the corresponding state, which we then continue annealing to the next temperature. In Section \ref{sec:naa} we describe the amplitude estimation algorithm of Brassard, Hoyer, Mosca, and Tapp (BHMT)~\cite{brassard} and demonstrate how the starting state can be restored at almost no additional cost in the number of Markov chain steps required. The nondestructive amplitude estimation algorithm can be summarized as follows:
\begin{theorem}[Nondestructive amplitude estimation]\label{thm:amplest1}
Given state $\ket{\psi}$ and reflections $R_\psi=2\ket{\psi}\bra{\psi}-I$ and $R=2P-I$, and any $\eta>0$, there exists a quantum algorithm that outputs $\tilde{a}$, an approximation to $a=\langle\psi|P|\psi\rangle$, so that
\[
|\tilde{a}-a|\leq 2\pi\frac{a(1-a)}{M}+\frac{\pi^2}{M^2}
\]
with probability at least $1-\eta$ and $O(\log(1/\eta)M)$ uses of $R_\psi$ and $R$. Moreover the algorithm restores the state $\ket{\psi}$ with probability at least $1-\eta$.
\end{theorem}
This is proved in \Cref{sec:naa}.

To perform amplitude estimation we will need to be able to perform the reflections $R_\psi=2\ket{\Pi_{\beta_i}}\bra{\Pi_{\beta_i}}-I$ and $R=2P-I=2\ket{\Pi_{\beta'}}\bra{\Pi_{\beta'}}-I$. The following theorem due to Magniez, Nayak, Roland, and Santha (MNRS)~\cite{magniez} allows us to approximate these reflections.
\begin{theorem}\label{thm:magniez}
(MNRS~\cite{magniez} Theorem 6) Suppose that we wish to approximate the reflection $R=2\ket{\Pi}\bra{\Pi}-I$ about $\ket{\Pi}$, where $\ket{\Pi}$ is the coherent encoding of $\Pi$, the stationary distribution of Markov chain $M$ with spectral gap $\delta$. Then there is a quantum circuit $\tilde{R}$ so that for $\ket{\Psi}$ orthogonal to $\ket{\Pi}$, $\|(\tilde{R}+I)\ket{\Psi}\|\leq 2^{1-k}$, and $\tilde{R}$ uses $O(k/\sqrt{\delta})$ steps of the quantum walk operator $W$ corresponding to $M$.
\end{theorem}

In the amplitude estimation algorithm we need to be able to perform $O(\log(1/\eta)M)$ applications of the Grover search operator $Q=-R_\psi R$ (see Section \ref{sec:naa} for more details), but instead we have access to an approximation $\tilde{Q}=-\tilde{R_\psi}\tilde{R}$. We claim that this error can be bounded using the following observation.
\begin{lemma}\label{lemma:induction}
Let $\tilde{R}_\psi$ and $\tilde{R}$ be the respective approximations to $R_\psi=2\ket{\psi}\bra{\psi}-I$ and $R=2P-I$ given by the algorithm of Theorem \ref{thm:magniez}. Then, letting $Q=-R_\psi R$ with approximation $\tilde{Q}=-\tilde{R}_\psi\tilde{R}$, and letting state $\ket{\psi'}\in\textup{span}\{\ket{\psi},\textup{Im}(P)\}$, the error in using approximate reflections can be bounded as $\|Q^i\ket{\psi'}-\tilde{Q}^i\ket{\psi'}\|\leq i2^{2-k}$.
\end{lemma}
\begin{proof}[Proof of Lemma \ref{lemma:induction}]
By induction. The $i=1$ case follows from Theorem \ref{thm:magniez}. Assume $\|Q^{i-1}\ket{\psi'}-\tilde{Q}^{i-1}\ket{\psi'}\|\leq (i-1)2^{2-k}$. Then $\|Q^i\ket{\psi'}-\tilde{Q}^i\ket{\psi'}\|\leq\|Q^{i-1}\ket{\psi'}-\tilde{Q}^{i-1}\ket{\psi'}\|+\|(Q-\tilde{Q})Q^i\ket{\psi'}\|\leq i2^{2-k}$.
\end{proof}

Using Lemma \ref{lemma:induction}, we can then restate the result on nondestructive amplitude estimation using approximate reflections.
\begin{theorem}\label{thm:amplest2}
(Nondestructive amplitude estimation using approximate reflections) Given state $\ket{\psi}$, an approximation $\tilde{R}_\psi$ to reflection $R_\psi=2\ket{\psi}\bra{\psi}-I$, an approximation $\tilde{R}$ to reflection $R=2P-I$, and any $\eta>0$, where all approximate reflections are given by Theorem \ref{thm:magniez}, there exists a quantum algorithm that outputs $\tilde{a}$, an approximation to $a=\langle\psi|P|\psi\rangle$, so that
\[
|\tilde{a}-a|\leq 2\pi a(1-a)\epsilon+\pi^2\epsilon^2
\]
with probability at least $1-\eta$. The algorithm restores the state $\ket{\psi}$ with probability at least $1-\eta$ and requires $O(1/(\epsilon\sqrt{\delta})\log(1/\epsilon)\log(1/\eta))$ steps of the quantum walk operators corresponding to $\tilde{R}_\psi$ and $\tilde{R}$, where $\delta$ lower bounds the spectral gaps of the corresponding Markov chains.
\end{theorem}
\begin{proof}[Proof of Theorem \ref{thm:amplest2}] The algorithm for nondestructive amplitude estimation (see Theorem \ref{thm:amplest1} and Section \ref{sec:naa}) requires the ability to generate the state $Q^M\ket{\psi}$. By Lemma \ref{lemma:induction} we know that we can generate an approximation $\tilde{Q}^M\ket{\psi}$ with $\|Q^M\ket{\psi}-\tilde{Q}^M\ket{\psi}\|\leq M2^{2-k}$. Taking $k=\log M+c$ then ensures that this error is bounded by a constant. Finally, calling $\epsilon=1/M$, we note that amplitude estimation occurs with error $O(\epsilon)$ if we require $O(1/\epsilon\log(1/\eta))$ uses of $\tilde{R}_\psi$ and $\tilde{R}$. With our choice of $k$, each use of $\tilde{R}_\psi$ and $\tilde{R}$ requires $O(\log(1/\epsilon)/\sqrt{\delta})$ Markov chain steps, so the algorithm requires $O(1/(\epsilon\sqrt{\delta})\log(1/\epsilon)\log(1/\eta))$ total Markov chain steps.
\end{proof}

Having described everything we need---the QSA algorithm, the binary search, and nondestructive amplitude estimation---we will now put everything together. The result will be two fully quantum algorithms, one for constructing an adaptive schedule and qsampling from the posterior distribution for Bayesian inference, and another for constructing an adaptive schedule and calculating $Z(\infty)$ for counting problems.

\subsection{QSA for Bayesian Inference}
The quantum algorithm for Bayesian inference is given by the following.
\begin{algorithm}[H]
\caption{QSA for Bayesian inference.}
\label{alg:bi}
\textbf{Input:} State $\ket{\Pi_0}=\sum_x \sqrt{\Pi_0}\ket{x}$, the coherent encoding of the prior distribution, constant $p>0$, and constant $\eta>0$. \\
\textbf{Output}: State $\ket{\widetilde{\Pi_1}}$, an approximation to the coherent encoding of the posterior distribution, and temperature schedule $\beta_0, \beta_1,\ldots,\beta_\ell$ with $\beta_0=0$ and $\beta_\ell=1$ so that $|\langle\Pi_{\beta_i}|\Pi_{
\beta_{i+1}}\rangle|\geq p$.
\begin{algorithmic}[1]
\FOR{i:=1 to $\ell=O\left(\sqrt{\bbE_{\Pi_0}[L(\theta)]\log(\bbE_{\Pi_0}[L(\theta)])}\right)$}
\STATE At current inverse temperature $\beta_i$ with state $\ket{\Pi_{\beta_i}}$,
\REPEAT 
\STATE Binary search on $\beta'\in[\beta_i, 1]$ with precision $1/(\max_\theta L(\theta))$.
\STATE Perform nondestructive amplitude estimation to calculate $|\langle\Pi_{\beta_i}|\Pi_{\beta'}\rangle|^2$ with error $\epsilon_e=p/10$ and failure probability $\eta/(\ell\max_\theta L(\theta))$.
\UNTIL{$|\langle\Pi_{\beta_i}|\Pi_{\beta'}\rangle|^2\geq p$.}
\STATE Anneal from $\ket{\Pi_{\beta_i}}$ to $\ket{\Pi_{\beta_{i+1}}}$ at inverse temperature $\beta_{i+1}=\beta'$.
\ENDFOR
\STATE Return $\ket{\Pi_{\beta_\ell}}$.
\end{algorithmic}
\end{algorithm}
For simplicity the above algorithm refers in each case to the ideal state, e.g. we write ``Return  $\ket{\Pi_{\beta_\ell}}$'' to mean that we return the state which approximates  $\ket{\Pi_{\beta_\ell}}$.

\begin{theorem}[Quantum adaptive annealing algorithm for Bayesian inference]\label{thm:qaaa}
Assume that we are given a prior distribution $\Pi_0(\theta)$ and a likelihood function $L(\theta)$, so that we can parametrize the partition function $Z(\beta)=\sum_\theta\Pi_0(\theta)e^{-\beta L(\theta)}$ at each inverse temperature $\beta\in[0,1]$. Assume that we can generate the state $\ket{\Pi_0}$ corresponding to the coherent encoding of the prior, and assume that for every inverse temperature $\beta$ we have a Markov chain $M_\beta$ with stationary distribution $\Pi_\beta$ and spectral gap lower-bounded by $\delta$. Then, for any $\epsilon>0$, $\eta>0$, there is a quantum algorithm that, with probability at least $1-\eta$, produces state $\ket{\widetilde{\Pi_1}}$ so that $\|\ket{\widetilde{\Pi_1}}-\ket{\Pi_1}\|\leq\epsilon$ for $\ket{\Pi_1}$ the coherent encoding of the posterior distribution $\Pi_1(\theta)=\Pi_0(\theta)e^{-L(\theta)}/Z(1)$. The algorithm uses
\begin{align*}
&O\left(\sqrt{\bbE_{\Pi_0}[L(\theta)]\log(\bbE_{\Pi_0}[L(\theta)])}\log^2(\sqrt{\bbE_{\Pi_0}[L(\theta)]\log(\bbE_{\Pi_0}[L(\theta)])}/(\epsilon\sqrt{\delta}))\log(\max_\theta L(\theta))\right.\\
&\qquad\left.\log(\sqrt{\bbE_{\Pi_0}[L(\theta)]\log(\bbE_{\Pi_0}[L(\theta)])}\max_\theta L(\theta)/\eta)\right)=\tilde{O}(\sqrt{\bbE_{\Pi_0}[L(\theta)]/\delta})
\end{align*}
total steps of the quantum walk operators corresponding to the Markov chains $M_\beta$.
\end{theorem}
\begin{proof}[Proof of Theorem \ref{thm:qaaa}] From Theorem \ref{thm:len} we know that the annealing schedule has length 
\[
\ell=O\left(\sqrt{\bbE_{\Pi_0}[L(\theta)]\log(\bbE_{\Pi_0}[L(\theta)])}\right).
\]
From Theorem \ref{thm:qsa} we know that, given a sequence of $\ell$ inverse temperatures $\{\beta_i\}$ with stationary distributions that satisfy $|\langle\Pi_{\beta_i}|\Pi_{\beta_{i+1}}\rangle|^2\geq p$ for a constant $p>0$, quantum annealing to the state at the final temperature $\beta_\ell$ takes
\[
O(\ell\delta^{-1/2}\log^2(\ell/\epsilon)(1/p)\log(1/p))
\]
total steps of the quantum walk operators corresponding to the $M_{\beta_i}$.

Since we simultaneously construct the schedule and anneal our state on the fly, we also need to account for the cost of constructing the schedule. At each inverse temperature $\beta_i$ we perform binary search to find inverse temperature $\beta_{i+1}$ satisfying $|\langle \Pi_{\beta_i}|\Pi_{\beta_{i+1}}\rangle|^2\geq p$ in the interval $[\beta_i,1]$. We choose binary search precision $1/(\max_\theta L(\theta))$ since $\Pi_{\beta}\propto e^{-\beta L(\theta)}$, which means that the binary search procedure contributes a factor of $\log(\max_\theta L(\theta))$ to the complexity. For each candidate inverse temperature $\beta'$ in the binary search, we perform nondestructive amplitude estimation to calculate $|\langle \Pi_{\beta_i}|\Pi_{\beta'}\rangle|^2$. We set the failure probability of nondestructive amplitude estimation to $\eta/(\ell\max_\theta L(\theta))$. From Theorem \ref{thm:amplest2}, we can estimate $|\langle \Pi_{\beta_i}|\Pi_{\beta'}\rangle|^2\geq p$ with error that is $O(\epsilon_e)$ using $O(1/(\epsilon_e\sqrt{\delta})\log(1/\epsilon_e)\log(\ell\max_\theta L(\theta)/\eta))$ Markov chain steps. Since we take $\epsilon_e=p/10$, our binary search then guarantees that we can find a sequence of $\ell$ inverse temperatures satisfying $|\langle \Pi_{\beta_i}|\Pi_{\beta_{i+1}}\rangle|^2\geq 9p/10$ with a total cost of
\[
O(\ell\delta^{-1/2}\log(\max_\theta L(\theta))(1/p)\log(1/p)\log(\ell(\max_\theta L(\theta))/\eta))
\]
total Markov chain steps. Adding the two contributions from constructing the schedule and annealing the state, we get a total cost of
\[
O(\ell\delta^{-1/2}\log(\max_\theta L(\theta))\log^2(\ell/\epsilon)(1/p)\log(1/p)\log(\ell(\max_\theta L(\theta))/\eta))=\tilde{O}(\sqrt{\bbE_{\Pi_0}[L(\theta)]/\delta})
\]
Markov chain steps.
\end{proof}

\subsection{QSA for Counting Problems}
In counting problems, we'd like to calculate $Z(\infty)$ according to the telescoping product given by (\ref{eq:telescope}), which means that we need to sample and estimate an expectation value $\bbE[W_{\beta_i, \beta_{i+1}}]$  at each inverse temperature $\beta_i$, where $W_{\beta_i, \beta_{i+1}}$ is given by equation (\ref{eq:w}). Computing the expectation value can be done using the amplitude estimation based algorithm of Montanaro, and moreover it can be made nondestructive using nondestructive amplitude estimation. This is Algorithm 4 of~\cite{montanaro}, which estimates an expectation value $\bbE(X)$ assuming bounded variance $\text{Var}(X)/(\bbE(X))^2\leq B$. Note that the bounded variance condition of Algorithm 4 is satisfied using the Chebyshev condition of Equation \ref{eq:boundedvariance}.

Thus we would expect a cost both in terms of the number of Markov chain steps required and in terms of the number of samples required, where the sample cost is incurred by the calculation of the expectation values, while the Markov chain cost is incurred both in the computation of the temperature schedule itself, and in the calculation of expectation values given the inverse temperatures.

The quantum algorithm for approximating $Z(\infty)$ is as follows. Note that in Line 5 we perform nondestructive amplitude estimation to determine the next temperature in the schedule, while in Lines 10, 13, and 16 we perform nondestructive amplitude estimation using Algorithm 4 of~\cite{montanaro} to estimate the $\bbE[W_{\beta_i, \beta_{i+1}}]$, which we then multiply together at the end to obtain an estimate for the partition function.  The proof of correctness of this algorithm is in \Cref{thm:counting}.

\begin{algorithm}[H]
\caption{QSA for computing partition functions for counting problems.}
\label{alg:count}
\textbf{Input:} Descriptions of state space $\Omega$ and energy function $H:\Omega\mapsto \mathbb{R}_+$.  Constant $B>0$, bound $n\geq \max_x H(x)$, error $\epsilon=O(1/\sqrt{\log\log |\Omega|})$, failure probability $\eta>0$, and $\tilde{O}(B\sqrt{\log |\Omega|}/\epsilon)$ copies of state $\ket{\Pi_0} := |\Omega|^{-1/2}\sum_{x\in\Omega}\ket x$. \\
\textbf{Output}: $\tilde{Z}$, an $\epsilon$-approximation to $Z(\infty)$, and $B$-Chebyshev cooling schedule $\beta_0=0, \beta_1,\ldots,\beta_\ell=\infty$ satisfying  $\log Z(\beta_{\ell-1})=1$.
\begin{algorithmic}[1]
\FOR{i:=1 to $O(\sqrt{\log|\Omega|\log(n)})$}
\STATE At current inverse temperature $\beta_i$ with states $\ket{\Pi_{\beta_i}}$,
\REPEAT 
\STATE Binary search on $\beta'\in[\beta_i, \gamma]$ with precision $1/n$.
\STATE Perform nondestructive amplitude estimation to estimate $|\langle\Pi_{\beta_i}|\Pi_{\beta'}\rangle|^2$ with error $\epsilon_e=p/10$ and failure probability $\eta/(n\log\log n\sqrt{\log|\Omega|\log(n)})$.
\UNTIL{our estimate satisfies $|\langle\Pi_{\beta_i}|\Pi_{\beta'}\rangle|^2\geq 1/B$.}
\STATE Set $\beta_{i+m+1}=(\beta_i+\beta')/2$ for $m=\lceil\log\log |\Omega|\rceil$.
\FOR{j:=1 to $m=\lceil\log\log |\Omega|\rceil$}
\STATE Set $\beta_{i+j}=\beta_i+(1-2^{-j})(\beta_{i+m+1}-\beta_i)$
\STATE Perform Algorithm 4 of~\cite{montanaro} on states $\ket{\Pi_{\beta_{i+j-1}}}$ using nondestructive amplitude estimation with error $\epsilon_e=\epsilon$ and failure probability $\eta/(n\log\log n\sqrt{\log|\Omega|\log(n)})$ to estimate $\bbE[W_{\beta_{i+j-1}, \beta_{i+j}}]$.
\STATE Anneal from states $\ket{\Pi_{\beta_{i+j-1}}}$ to states $\ket{\Pi_{\beta_{i+j}}}$.
\ENDFOR
\STATE Perform Algorithm 4 of~\cite{montanaro} on states $\ket{\Pi_{\beta_{i+m}}}$ using nondestructive amplitude estimation with error $\epsilon_e=\epsilon$ and failure probability $\eta/(n\log\log n\sqrt{\log|\Omega|\log(n)})$ to estimate $\bbE[W_{\beta_{i+m}, \beta_{i+m+1}}]$.
\STATE Anneal from states $\ket{\Pi_{\beta_{i+m}}}$ to states $\ket{\Pi_{\beta_{i+m+1}}}$.
\ENDFOR
\STATE Perform Algorithm 4 of~\cite{montanaro} on states $\ket{\Pi_{\gamma}}$ using nondestructive amplitude estimation with error $\epsilon_e=\epsilon$ and failure probability $\eta/(n\log\log n\sqrt{\log|\Omega|\log(n)})$ to estimate $\bbE[W_{\gamma, \infty}]$.
\STATE Return $\tilde{Z}=\prod_{i=0}^{\ell-1}\bbE[W_{\beta_i, \beta_{i+1}}]$
\end{algorithmic}
\end{algorithm}
As in \Cref{alg:bi} we use notation that ignores the errors in our estimates.  Specifically, in the last line we write $\bbE[W_{\beta_i, \beta_{i+1}}]$ to mean our estimates of this that we have computed in lines 10, 13, and 16.  Likewise we refer to various states $\ket{\Pi_\beta}$ while our algorithm actually has access to approximate versions of those states.

The following theorem due to Montanaro~\cite{montanaro} specifies how many total qsamples are needed to calculate $Z(\infty)$.

\begin{theorem}[Montanaro~\cite{montanaro} Theorem 8]
Given a counting problem partition function $Z(\beta)$ and a $B$-Chebyshev cooling schedule $\beta_0, \beta_1,\ldots,\beta_\ell$ with $\beta_0=0$ and $\beta_\ell=\infty$, and assuming the ability to qsample from each Gibbs distribution $\Pi_{\beta_i}$, there is a quantum algorithm which outputs an estimate $\tilde{Z}$ of $Z(\infty)$ such that
\[
\Pr\left[ (1-\epsilon)Z(\infty)\leq\tilde{Z}\leq(1+\epsilon)Z(\infty)\right]\leq 3/4
\]
using
\[
O\left(\frac{B\ell\log\ell}{\epsilon}\log^{3/2}\left(\frac{B\ell}{\epsilon}\right)\log\log\left(\frac{B\ell}{\epsilon}\right)\right)
\]
qsamples at each $\Pi_{\beta_i}$, which corresponds to
\[
O\left(\frac{B\ell^2\log\ell}{\epsilon}\log^{3/2}\left(\frac{B\ell}{\epsilon}\right)\log\log\left(\frac{B\ell}{\epsilon}\right)\right)=\tilde{O}(B\ell^2/\epsilon)
\]
qsamples in total.
\end{theorem}

The cost in terms of quantum walk steps needed can be split up into two parts: the cost of computing the schedule itself (that is,  determining the inverse temperatures and annealing through them), and the cost of computing the expectation values in order to estimate $Z(\infty)$ (that is, given each inverse temperature). The following theorem due to Montanaro~\cite{montanaro} specifies the total Markov chain steps needed to estimate the expectation values given a temperature schedule.

\begin{theorem}
(Montanaro~\cite{montanaro} Theorem 11) Given a counting problem partition function $Z(\beta)$, a $B$-Chebyshev cooling schedule $\beta_0, \beta_1,\ldots,\beta_\ell$ with $\beta_0=0$ and $\beta_\ell=\infty$, and a series of Markov chains with stationary distributions $\Pi_{\beta_i}$ and spectral gap lower bounded by $\delta$, and assuming the ability to qsample from $\Pi_0$, for any $\eta>0$ and $\epsilon=O(1/\sqrt{\log\ell})$ there exists a quantum algorithm which uses
\[
O((\ell^2/\sqrt{\delta}\epsilon)\log^{5/2}(\ell/\epsilon)\log(\ell/\eta)\log\log(\ell/\epsilon))=\tilde{O}(\ell^2/\sqrt{\delta}\epsilon)
\]
steps of the quantum walk operators corresponding to the Markov chains and outputs $\tilde{Z}$, an estimate of $Z(\infty)$ such that
\[
\Pr\left[ (1-\epsilon)Z(\infty)\leq\tilde{Z}\leq(1+\epsilon)Z(\infty)\right]\geq 1-\eta.
\]
\end{theorem}

We claim that analogous to the case of Bayesian inference, the construction of the schedule itself can be completed with $\tilde{O}(\sqrt{\ell/\delta})$ Markov chain steps:
\begin{theorem}\label{thm:constructcount}
Given the counting problem partition function $Z(\beta)=\sum_{k=0}^na_ke^{-\beta k}$ with $|\Omega|=\sum_{k=0}^n a_k$ and $n=\max_x H(x)$, assume that we can generate state $\ket{\Pi_0}$ corresponding to the uniform distribution over $\Omega$. Letting $\gamma$ be the temperature at which $\log Z(\gamma)=1$, assume also that for every inverse temperature $\beta\in[0, \gamma]$ we have a Markov chain $M_\beta$ with stationary distribution $\Pi_\beta$ and spectral gap lower-bounded by $\delta$. Then, for any $\epsilon>0$, $\eta>0$, there is a quantum algorithm (\Cref{alg:count}, lines 1--15) which anneals through the sequence of states $\ket{\widetilde{\Pi_{\beta_i}}}$ so that $\|\ket{\widetilde{\Pi_{\beta_i}}}-\ket{\Pi_{\beta_i}}\|\leq\epsilon$ for $\ket{\Pi_{\beta_i}}$ the coherent encoding of the Gibbs distribution at inverse temperatures $\beta_i$. The algorithm uses
\[
O\left(\log\log |\Omega|\sqrt{\log|\Omega|\log(n)}\delta^{-1/2}\log^2(\sqrt{\log|\Omega|\log(n)}/\epsilon)\log n\log(n\log\log n\sqrt{\log|\Omega|\log(n)}/\eta)\right)=\tilde{O}(\sqrt{(\log |\Omega|)/\delta})
\]
total steps of the quantum walk operators corresponding to the Markov chains $M_\beta$.
\end{theorem}
\begin{proof}[Proof of Theorem \ref{thm:constructcount}] According to Theorem \ref{thm:countinglen}, the $B$-Chebyshev cooling schedule has length 
\[
\ell=O\left(\log\log |\Omega|\sqrt{\log|\Omega|\log(n)}\right).
\]
From Theorem \ref{thm:qsa} we know that, given a sequence of $\ell$ inverse temperatures $\{\beta_i\}$ with stationary distributions that satisfy $|\langle\Pi_{\beta_i}|\Pi_{\beta_{i+1}}\rangle|^2\geq p$ for a constant $p>0$, quantum annealing through the sequence of states $\ket{\Pi_{\beta_i}}$ corresponding to the Gibbs distributions $\Pi_{\beta_i}$ at each inverse temperature $\beta_i$ takes
\[
O(\ell\delta^{-1/2}\log^2(\ell/\epsilon)(1/p)\log(1/p))
\]
total steps of the quantum walk operators corresponding to the $M_{\beta_i}$. Note that here, unlike in the case of Bayesian inference, we need to show that $|\langle\Pi_{\beta_i}|\Pi_{\beta_{i+1}}\rangle|^2\geq p$ is satisfied as the $B$-Chebyshev condition instead guarantees that $|\langle\Pi_{\beta_i}|\Pi_{2\beta_{i+1}-\beta}\rangle|^2\geq 1/B$ is satisfied. But we claim that satisfying the latter is enough to satisfy the former. To see this, note that $|\langle\Pi_{\beta_i}|\Pi_{2\beta_{i+1}-\beta}\rangle|^2\geq p$ is equivalent to the $B$-Chebyshev condition with $B=1/p$, and that the $B$-Chebyshev condition can be rewritten as
\begin{equation}\label{eq:chebyshevequiv}
\frac{Z(\beta_i)Z(2\beta_{i+1}-\beta_i)}{Z(\beta_{i+1})^2}=\sum_{x\in\Omega}\frac{\Pi_{\beta_{i+1}}(x)^2}{\Pi_{\beta_i}(x)}\geq\frac{1}{p}.
\end{equation}
The overlap $\langle\Pi_{\beta_i}|\Pi_{\beta_{i+1}}\rangle$ which appears in the slow-varying condition can be rewritten as
\begin{align}
\langle\Pi_{\beta_i}|\Pi_{\beta_{i+1}}\rangle&=\sum_{x\in\Omega}\Pi_{\beta_{i+1}}(x)\sqrt{\frac{\Pi_{\beta_i}(x)}{\Pi_{\beta_{i+1}}(x)}}\nonumber\\
&\geq\frac{1}{\sqrt{\sum_{x\in\Omega}\Pi_{\beta_{i+1}}(x)\frac{\Pi_{\beta_{i+1}}(x)}{\Pi_{\beta_i}(x)}}}=\sqrt{p}\label{eq:chebyshevequiv2}
\end{align}
where we obtain the inequality from Jensen's inequality in the form $1/\sqrt{\bbE[X]}\leq \bbE[1/\sqrt{X}]$.

Since we simultaneously construct the schedule and anneal our state on the fly, we also need to account for the cost of constructing the schedule. At each inverse temperature $\beta_i$ we perform binary search to find temperature $2\beta_{i+1}-\beta_i$ satisfying $|\langle \Pi_{\beta_i}|\Pi_{2\beta_{i+1}-\beta_i}\rangle|^2\geq p$ in the interval $[\beta_i,\gamma]$. We choose binary search precision $1/n$ since $\Pi_{\beta}\propto e^{-\beta k}$ for $k\in\{0,n\}$, which means that the binary search procedure contributes a factor of $\log n$ to the complexity. For each candidate inverse temperature $\beta'$ in the binary search, we perform nondestructive amplitude estimation to calculate $|\langle \Pi_{\beta_i}|\Pi_{\beta'}\rangle|^2$. We set the failure probability of amplitude estimation to be $\eta/(n\ell)$. From Theorem \ref{thm:amplest2}, we can estimate $|\langle \Pi_{\beta_i}|\Pi_{\beta'}\rangle|^2\geq p$ with error that is $O(\epsilon_e)$ using $O(1/\epsilon_e\log(1/\epsilon_e)\log(n\ell/\eta)/\sqrt{\delta})$ Markov chain steps. Since we take $\epsilon_e=p/10$, our binary search then guarantees that we can find a sequence of $\ell$ inverse temperatures satisfying both $|\langle \Pi_{\beta_i}|\Pi_{2\beta_{i+1}-\beta_i}\rangle|^2\geq 9p/10$ and $|\langle \Pi_{\beta_i}|\Pi_{\beta_{i+1}}\rangle|^2\geq 9p/10$ with a total cost of
\[
O(\ell\delta^{-1/2}\log n(1/p)\log(1/p)\log(n\ell/\eta))
\]
total Markov chain steps. Adding the two contributions from constructing the schedule and annealing the state, we get a total cost of
\[
O(\ell\delta^{-1/2}\log n\log^2(\ell/\epsilon)(1/p)\log(1/p)\log(n\ell/\eta))=\tilde{O}(\sqrt{(\log |\Omega|)/\delta})
\]
Markov chain steps.
\end{proof}

Adding these two contributions to the total number of Markov chain steps required (and noting that $O(B\ell/\epsilon)$ samples are needed at each of the $\ell$ inverse temperatures), we get a total complexity of $\tilde{O}((\log |\Omega|)/\sqrt{\delta}\epsilon)$, where $\epsilon$ is the error in computing the approximation to $Z(\infty)$. Thus we can finally state the following for the counting problem:
\begin{theorem}[Quantum adaptive annealing algorithm for computing partition functions for counting problems]\label{thm:counting}
 Given the counting problem partition function $Z(\beta)=\sum_{k=0}^na_ke^{-\beta k}$ with $|\Omega|=\sum_{k=0}^n a_k$ and $n=\max_x H(x)$, assume that we can generate state $\ket{\Pi_0}$ corresponding to the uniform distribution over $\Omega$. Letting $\gamma$ be the temperature at which $\log Z(\gamma)=1$, assume also that for every inverse temperature $\beta\in[0, \gamma]$ we have a Markov chain $M_\beta$ with stationary distribution $\Pi_\beta$ and spectral gap lower-bounded by $\delta$. Then, for any $\epsilon=O(1/\sqrt{\log\log |\Omega|})$ and any $\eta>0$, there is a quantum algorithm (\Cref{alg:count}) that uses 
\[
O((\ell^2/\sqrt{\delta}\epsilon)\log^{5/2}(\ell/\epsilon)\log(\ell/\eta)\log\log(\ell/\epsilon))+O((\ell/\epsilon)\ell\delta^{-1/2}\log n\log^2(\ell/\epsilon)\log(n\ell/\eta))=\tilde{O}((\log |\Omega|)/\sqrt{\delta}\epsilon)
\]
steps of the Markov chains and outputs $\tilde{Z}$, an approximation to $Z(\infty)$ such that
\[
\Pr\left[(1-\epsilon)Z(\infty)\leq\tilde{Z}\leq(1+\epsilon)Z(\infty)\right]\geq 1-\eta.
\]
\end{theorem}
In \Cref{subsec:partitionfns} we give several examples of partition function problems, and we evaluate the runtime of our algorithm on these examples.

\section{Nondestructive Amplitude Estimation}\label{sec:naa}
In this section we first describe the amplitude estimation algorithm of Brassard, Hoyer, Mosca, and Tapp (BHMT)~\cite{brassard}, and then we show how it can be made nondestructive. The result of BHMT can be stated as follows:
\begin{theorem}
(BHMT~\cite{brassard} Theorem 12) Given state $\ket{\psi}$ and reflections $R_\psi=2\ket{\psi}\bra{\psi}-I$ and $R=2P-I$, there exists a quantum algorithm that outputs $\tilde{a}$, an approximation to $a=\langle\psi|P|\psi\rangle$, so that
\[
|\tilde{a}-a|\leq 2\pi\frac{a(1-a)}{M}+\frac{\pi^2}{M^2}
\]
with probability at least $8/\pi^2$ and $M$ uses of $R_\psi$ and $R$.
\end{theorem}
In amplitude estimation we are interested in the eigenspectrum of the Grover search operator, given by
\begin{equation}
Q=-R_\psi R.
\end{equation}
We can decompose our original Hilbert space into ${\cal H}_1=\text{Im}(P)$ and its complement ${\cal H}_0$. Writing $\ket{\psi}$ as
\begin{equation}
\ket{\psi}=\sin\theta\ket{\psi_1}+\cos\theta\ket{\psi_0}
\end{equation}
for $\ket{\psi_1}\in{\cal H}_1$ and $\ket{\psi_0}\in{\cal H}_0$, we note that on the space spanned by $\{\ket{\psi_1}, \ket{\psi_0}\}$, $Q$ acts as
\begin{equation}\label{eq:qmatrix}
Q=\left(
\begin{array}{cc}
\cos(2\theta) & \sin(2\theta)\\
-\sin(2\theta) & \cos(2\theta)
\end{array}
\right).
\end{equation}
This matrix has eigenvalues $e^{\pm 2i\theta}$ with corresponding eigenvectors
\begin{equation}
\ket{\psi_\pm}=\frac{1}{\sqrt{2}}(\ket{\psi_1}\pm i\ket{\psi_0}).
\end{equation}
Since $a=\langle\psi|P|\psi\rangle=\sin^2\theta$, estimating the eigenvalues of $Q$ allows us to estimate $a$. To estimate the eigenvalues of $Q$, BHMT define the Fourier transform
\begin{equation}
F_M:\ket{x}\mapsto\frac{1}{\sqrt{M}}\sum_{y=0}^{M-1}e^{2\pi ixy/M}\ket{y}
\end{equation}
and the state
\begin{equation}
\ket{S_M(\omega)}=\frac{1}{\sqrt{M}}\sum_{y=0}^{M-1}e^{2\pi i\omega y}\ket{y}.
\end{equation}
Then performing $F_M^{-1}\ket{S_M(\omega)}$ and measuring in the computational basis allows us to perform phase estimation. Explicitly, according to BHMT Theorem 11,
\begin{theorem}\label{thm:phasest}
(Phase estimation, BHMT Theorem 11) Let $y$ be the random variable corresponding to the result of measuring $F_M^{-1}\ket{S_M(\omega)}$. If $M\omega$ is an integer, then $\Pr\left[y=M\omega\right]=1$. Otherwise,
\begin{equation}\label{eq:phasest}
P\left(\left|\frac{y}{M}-\omega\right|\leq \frac{1}{M}\right)\geq \frac{8}{\pi^2}.
\end{equation}
\end{theorem}

\noindent Finally, we will also need to define the operator
\begin{equation}
\Lambda_M(U):\ket{j}\ket{y}\mapsto\ket{j}U^j\ket{y}.
\end{equation}
Now we can state the amplitude estimation algorithm:

\begin{algorithm}[H]
\caption{Amplitude estimation algorithm.}
\label{alg:amplest}
\textbf{Input}: State $\ket{\psi}$ and operators $R_\psi=2\ket{\psi}\bra{\psi}-I$ and $R=2P-I$. \\
\textbf{Output}: $\tilde{a}$, an estimate of $\langle \psi | P | \psi\rangle$.
\begin{algorithmic}[1]
\STATE Start with state $\ket{0}\ket{\psi}$.
\STATE Apply operator $(F_M^{-1}\otimes I)\Lambda_M(Q)(F_M\otimes I)$.
\STATE Measure first register to obtain either state $\ket{y}\ket{\psi_+}$ or $\ket{y}\ket{\psi_-}$.
\STATE Return $\tilde{a}=\sin^2(\pi y/M)$.
\end{algorithmic}
\end{algorithm}

We can boost the success probability of amplitude estimation using the powering lemma~\cite{powering}, which improves the amplitude estimation success probability of $8/\pi^2$ to $1-\eta$ for any $\eta>0$ at the cost of an extra $O(\log 1/\eta)$ factor.

\begin{lemma}
(Powering lemma~\cite{powering}) Suppose we have an algorithm that produces an estimate $\tilde{\mu}$ of $\mu$ so that $|\mu-\tilde{\mu}|<\epsilon$ with some fixed probability $p>1/2$. Then for any $\eta>0$, repeating the algorithm $O(\log 1/\eta)$ times and taking the median suffices to produce $\tilde{\mu}$ with $|\mu-\tilde{\mu}|<\epsilon$ with probability at least $1-\eta$.
\end{lemma}

\noindent This allows us to state the following version of amplitude estimation with powering:
\begin{algorithm}[H]
\caption{Amplitude estimation with powering.}
\label{alg:amplestpower}
\textbf{Input}: State $\ket{\psi}$, operators $R_\psi=2\ket{\psi}\bra{\psi}-I$ and $R=2P-I$, and $\eta>0$. \\
\textbf{Output}: $\tilde{a}$, an estimate of $\langle \psi | P | \psi\rangle$.
\begin{algorithmic}[1]
\STATE Start with state $\ket{\psi}$.
\FOR{i:=1 to $q=O(\log(1/\eta))$}
\STATE Add a new register $\ket{0}_i$.
\STATE Apply operator $(F_M^{-1}\otimes I)\Lambda_M(Q)(F_M\otimes I)$ on subsystem $\ket{0}_i\ket{\psi}$.
\ENDFOR
\STATE Add register $\ket{0}_{q+1}$ and apply the function that maps the median of the first $q$ registers to this register.
\STATE Uncompute the first $q$ registers.
\STATE Measure $(q+1)$-st register to obtain median $y_m$.
\STATE Return $\tilde{a}=\sin^2(\pi y_m/M)$.
\end{algorithmic}
\end{algorithm}

After performing amplitude estimation, we'd like to restore our state to the initial starting state. To do so, we start by observing that we can rewrite the state $\ket{\psi}$ as
\begin{equation}
\ket{\psi}=\frac{1}{\sqrt{2}}(e^{-i\theta}\ket{\psi_+}+e^{i\theta}\ket{\psi_-}).
\end{equation}
Then applying the operator of step 2 of Algorithm \ref{alg:amplest} yields the following sequence of states:
\begin{align*}
((F_M^{-1}\otimes I)\Lambda_M(Q)(F_M\otimes I))\ket{0}\ket{\psi}&=((F_M^{-1}\otimes I)\Lambda_M(Q)(F_M\otimes I))\left(\frac{1}{\sqrt{2}}\ket{0} (e^{-i\theta}\ket{\psi_+}+e^{i\theta}\ket{\psi_-})\right)\\
&=((F_M^{-1}\otimes I)\Lambda_M(Q))\left(\frac{1}{\sqrt{2M}}\sum_{j=0}^{M-1}\ket{j}(e^{-i\theta}\ket{\psi_+}+e^{i\theta}\ket{\psi_-})\right)\\
&=(F_M^{-1}\otimes I)\left(\frac{e^{-i\theta}}{\sqrt{2M}}\sum_{j=0}^{M-1}e^{2ij\theta}\ket{j}\ket{\psi_+}+\frac{e^{i\theta}}{\sqrt{2M}}\sum_{j=0}^{M-1}e^{-2ij\theta}\ket{j}\ket{\psi_-}\right)\\
&=\frac{e^{-i\theta}}{\sqrt{2}}(F_M^{-1}\ket{S_M(\theta/\pi)})\ket{\psi_+}+\frac{e^{i\theta}}{\sqrt{2}}(F_M^{-1}\ket{S_M(1-\theta/\pi)})\ket{\psi_-}
\end{align*}
Thus after the measurement in step 3, the algorithm will always end in either of the two states $\ket{j}\ket{\psi_{\pm}}$.

Note that we'd like to restore this to the starting state $\ket{0}\ket{\psi}$, and that $|\langle\psi|\psi_\pm\rangle|^2=1/2$ is a constant. Since this overlap is constant, and since we are working with two-dimensional subspaces, we can restore the state using a scheme similar to that of Temme et. al.~\cite{temme}, which was in turn inspired by a scheme of Marriott and Watrous~\cite{qma}.\footnote{We thank Fernando Brand\~ao for discussions related to this point.} That is, given $\ket{\psi_{\pm}}$, we first apply the projection operator $\ket{\psi}\bra{\psi}=(R_{\psi}+I)/2$. We either obtain $\ket{\psi}$, in which case we are done, or we obtain some $\ket{\psi^{\perp}}$ so that $\langle\psi|\psi^{\perp}\rangle=0$. Since $\ket{\psi^{\perp}}$ can also be expressed in the basis $\{\ket{\psi_+},\ket{\psi_-}\}$, we can again apply amplitude estimation to collapse the last register onto either $\ket{\psi_+}$ or $\ket{\psi_-}$. Then we repeat the projection onto $\ket{\psi}$. Since the overlap between $\ket{\psi}$ and $\ket{\psi_{\pm}}$ is constant, the expected numbers of times we need to perform the series of projections before attaining our desired state $\ket{\psi}$ is constant as well.

This suggests the following algorithm for state restoration:
\begin{algorithm}[H]
\caption{State restoration following amplitude estimation.}
\label{alg:staterest}
\textbf{Input}: $\eta>0$; either state $\ket{\psi_+}$ or $\ket{\psi_-}$; and operators $R_\psi=2\ket{\psi}\bra{\psi}-I$ and $R=2P-I$, where $\ket{\psi_\pm}$ are the eigenstates of $Q=-R_{\psi}R$ with eigenvalues $e^{\pm 2i\theta}$. \\
\textbf{Output}: State $\ket{\psi}$.
\begin{algorithmic}[1]
\WHILE{current state is not $\ket{\psi}$}
\STATE Apply $(R_{\psi}+I)/2$.
\IF{current state is $\ket{\psi}$}
\STATE Return $\ket{\psi}$.
\ENDIF
\FOR{i:=1 to $q=O(\log(1/\eta))$}
\STATE Add a new register $\ket{0}_i$.
\STATE Apply operator $(F_M^{-1}\otimes I)\Lambda_M(Q)(F_M\otimes I)$ on subsystem $\ket{0}_i\ket{\psi}$.
\ENDFOR
\STATE Add register $\ket{0}_{q+1}$ and apply the function that maps the median of the first $q$ registers to this register.
\STATE Uncompute the first $q$ registers.
\STATE Measure $(q+1)$-st register to obtain either $\ket{\psi_+}$ or $\ket{\psi_-}$.
\ENDWHILE
\end{algorithmic}
\end{algorithm}

\noindent Performing amplitude estimation according to Algorithm \ref{alg:amplestpower} with failure probability less than $\eta/2$, and then performing state restoration according to Algorithm \ref{alg:staterest} with failure probability less than $\eta/2$, gives us an algorithm for nondestructive amplitude estimation with probability of success at least $1-\eta$:
 
\begin{theorem}
(Nondestructive amplitude estimation) Given state $\ket{\psi}$ and reflections $R_\psi=2\ket{\psi}\bra{\psi}-I$ and $R=2P-I$, and any $\eta>0$, there exists a quantum algorithm that outputs $\tilde{a}$, an approximation to $a=\langle\psi|P|\psi\rangle$, so that
\[
|\tilde{a}-a|\leq 2\pi\frac{a(1-a)}{M}+\frac{\pi^2}{M^2}
\]
with probability at least $1-\eta$ and $O(\log(1/\eta)M)$ uses of $R_\psi$ and $R$. Moreover the algorithm restores the state $\ket{\psi}$ with probability at least $1-\eta$.
\end{theorem}

\section{Discussion and Applications}\label{sec:discuss}

\subsection{Applications to Partition Function Problems}\label{subsec:partitionfns}
In this section, following the treatment of~\cite{montanaro} and~\cite{stefankovic}, we give several examples of problems from statistical physics and computer science that can be framed as partition function problems. We then show how our algorithm can be applied to obtain a speedup. We obtain a quadratic improvement in the scaling with $\epsilon$ due to Montanaro's algorithm for computing expectation values~\cite{montanaro}, and we obtain an improvement in the scaling with graph parameters due to the adaptive schedule of~\cite{stefankovic} and the QSA algorithm of~\cite{wocjan}. The results are summarized in Table \ref{tab:examples} and elaborated below.

\paragraph{Counting $k$-colorings} In the $k$-coloring problem, we are given a graph $G=(V,E)$ with maximum degree $\Delta$, and we'd like to count the number of ways to color the vertices with $k$ colors such that no two adjacent vertices share the same color (in statistical physics, this problem is also known as the antiferromagnetic Potts model at zero temperature). Here $\Omega$ is the set of colorings of $G$, and for each $\sigma\in\Omega$, $H(\sigma)$ is the number of monochromatic edges in $\sigma$. Thus we have the partition function
\[
Z(\beta)=\sum_{\sigma\in\Omega}e^{-\beta H(\sigma)}.
\]
We know that $|\Omega|=Z(0)=k^{|V|}$, and we'd like to calculate $Z(\infty)$, corresponding to the number of valid $k$-colorings. Jerrum~\cite{jerrum1} showed that using Glauber dynamics, a single site update Markov chain, it is possible to obtain mixing time $O(|V|\log |V|)$ whenever $k>2\Delta$. Thus our quantum algorithm can obtain an approximation for the $k$-coloring problem in time $\tilde{O}(|V|^{3/2}/\epsilon)$, whereas the classical algorithm of SVV scales like $\tilde{O}(|V|^2/\epsilon^2)$, and the partially quantum algorithm of Montanaro scales like $\tilde{O}(|V|^{3/2}/\epsilon+|V|^2)$.

\paragraph{Ising Model} The Ising model on a graph $G=(V, E)$ is a model from statistical physics where we place a spin at each vertex and assign each spin a value of $+1$ or $-1$. The Hamiltonian counts the number of edges whose endpoints have different spins. Here the space of possible assignments is given by $\Omega=\{\pm 1\}^{|V|}$, so $|\Omega|=Z(0)=2^{|V|}$. The Ising model has been extensively studied, and results such as~\cite{martinelli, mossel} show that in certain regimes, Glauber dynamics mixes rapidly, in time $O(|V|\log |V|)$. Thus our quantum algorithm scales like $\tilde{O}(|V|^{3/2}/\epsilon)$, while the classical algorithm of SVV~\cite{stefankovic} scales like $\tilde{O}(|V|^2/\epsilon^2)$, and the partially quantum algorithm of Montanaro~\cite{montanaro} scales like $\tilde{O}(|V|^{3/2}/\epsilon+|V|^2)$.

\paragraph{Counting Matchings}
A matching over a graph $G=(V, E)$ is a subset of edges that share no vertex in common. Letting $\Omega$ denote the set of all matchings over $G$, we then have a partition function of the form
\[
Z(\beta)=\sum_{M\in\Omega}e^{-\beta|M|}.
\]
Then we know that $Z(\infty)=1$, and we seek to calculate $Z(0)=|\Omega|$. Here we would need to anneal backwards in temperature; that is, if we had inverse temperatures $\beta_0=0<\beta_1<\ldots<\beta_\ell=\infty$, we would want to anneal in the reverse order,
\[
Z(0)=Z(\infty)\frac{Z(\beta_{\ell-1})}{Z(\infty)}\frac{Z(\beta_{\ell-2})}{Z(\beta_{\ell-1})}\cdots\frac{Z(0)}{Z(\beta_1)}.
\]
We would want to satisfy the Chebyshev condition in reverse as well; that is, we'd like to have
\[
\frac{Z(2\beta_i-\beta_{i+1})Z(\beta_{i+1})}{Z(\beta_i)^2}\leq B
\]
Note that as in the case of the non-reversed schedule, we take $\beta_{\ell-1}=\gamma_0$ so that $Z(\gamma_0)=e$ in order to satisfy the Chebyshev condition between $\beta_{\ell-1}$ and $\beta_\ell=\infty$. Next we need to anneal backwards from $\beta=\gamma_0$ to $\beta=0$. To do this we will modify the partition function to
\[
Z(\beta')=\sum_{x\in\Omega}e^{(\beta'-\gamma_0)H(x)}
\]
and anneal forwards from $\beta'=0$, corresponding to $Z(\beta'=0)=Z(\beta=\gamma_0)=e$, to $\beta'=\gamma_0$, corresponding to $Z(\beta'=\gamma_0)=Z(\beta=0)=|\Omega|$. Since $Z(\beta')$ is still a convex function, the results from Appendix \ref{app:schedlen} and Section \ref{sec:existencesched} guaranteeing the existence of a quadratically shorter schedule satisfying the Chebyshev condition still apply. (Note that the original paper by SVV~\cite{stefankovic} showed the existence of this cooling schedule for $\log Z(\beta)$ a decreasing function, but the argument in Appendix \ref{app:schedlen} applies equally well to increasing convex functions.)

Jerrum and Sinclair~\cite{js} showed that the Markov chain for computing matchings has mixing time $O(|V\|E|)$. Since $|\Omega|=O(|V|!\cdot 2^{|V|})$, our quantum algorithm has complexity $\tilde{O}(|V|^{3/2}|E|^{1/2}/\epsilon)$, compared to the $\tilde{O}(|V|^2|E|/\epsilon^2)$ complexity of SVV~\cite{stefankovic} and the $\tilde{O}(|V|^{3/2}|E|^{1/2}/\epsilon+|V|^2|E|)$ complexity of Montanaro~\cite{montanaro}.

\paragraph{Counting Independent Sets} An independent set on a graph $G=(V, E)$ with maximum degree $\Delta$ is a set of vertices that share no edge. Letting $\Omega$ denote the set of independent sets on $G$, and given a fugacity $\lambda>0$, we define
\[
Z(\beta)=\sum_{\sigma\in\Omega}\lambda^{|\sigma|}.
\]
Again we know that $Z(\infty)=1$, and we seek to calculate $Z(0)=|\Omega|$. As with the case of counting matchings, we can anneal backwards by modifying the partition function.

Vigoda~\cite{vigoda} showed that Glauber dynamics results in a mixing time of $O(|V|\log |V|)$ whenever $\lambda<2/(\Delta-2)$. Since $|\Omega|=O(2^{|V|})$, our quantum algorithm has complexity $\tilde{O}(|V|^{3/2}/\epsilon)$, while the classical algorithm of SVV~\cite{stefankovic} scales like $\tilde{O}(|V|^2/\epsilon^2)$, and the algorithm of Montanaro~\cite{montanaro} scales like $\tilde{O}(|V|^{3/2}/\epsilon+|V|^2)$.

\subsection{Warm Starts and Nonadaptive Schedules}\label{subsec:warm-starts}
Montanaro's quantum algorithm~\cite{montanaro} is already a sort of quantum version of SVV~\cite{stefankovic}.  So why doesn't it already achieve what we do?  Montanaro cites two related obstacles: warm starts and nonadaptive schedules.  In this section we will explain how warm starts are used by SVV, and why SVV use nonadaptive schedules to construct a schedule with warm starts.  For SVV this choice was not strictly necessary, but rather due to the fact that they consider applications to counting problems, where there is almost no additional cost to using nonadaptive schedules to ensure warm starts.  In the quantum case warm starts are still desirable, but achieving them using nonadaptive schedules is too costly, especially without the nondestructive amplitude estimation that we introduced in Section \ref{sec:naa}.  This led Montanaro to develop an algorithm that still relied on SVV's classical algorithm to construct a schedule with warm starts, and then used this schedule as input to the quantum walks.

We now explain these points in more detail.

\paragraph{Warm starts.}
The idea behind warm starts for classical random walks is that the spectral gap (directly) controls convergence in the 2-norm while applications usually require bounds in the 1-norm.  This norm conversion introduces some cost which is greatly reduced by starting the random walk in a distribution that is close to the target distribution, aka a ``warm start.''

To make this more concrete, we define two notions of distance between probability distributions. The total variation distance is
\[
\|\Pi_1-\Pi_2\|_{TV}=\frac{1}{2}\sum_{x\in\Omega}|\Pi_1(x)-\Pi_2(x)|
\]
and the $L^2$ distance, which is also a variance, is
\begin{align*}
\left\|\frac{\Pi_1}{\Pi_2}-1\right\|^2_{2, \Pi_2}&=\text{Var}_{\Pi_2}(\Pi_1/\Pi_2)\\
&=\sum_{x\in\Omega}\Pi_2(x)\left(\frac{\Pi_1(x)}{\Pi_2(x)}-1\right)^2.
\end{align*}
Now consider a Markov chain with stationary distribution $\Pi$, and suppose that we run this Markov chain on a starting distribution $\nu_0$ for $t$ steps to obtain distribution $\nu_t$. Letting $\delta$ be the spectral gap of the Markov chain, we have
\begin{equation}
\|\nu_t-\Pi\|_{TV}\leq e^{-\delta t/2}\left\|\frac{\nu_0}{\Pi}-1\right\|_{2,\Pi}
\label{eq:TV-from-2}\end{equation}
(see, for example, SVV~\cite{stefankovic} Lemma 7.3). In particular, the idea behind warm starts is to pick a warm start distribution $\nu_0$ so that the variance $\left\|\frac{\nu_0}{\Pi}-1\right\|_{2,\Pi}$ is bounded.
A ``cold start'', on the other hand, would be a choice of  $\nu_0$ that is far from $\Pi$, such as putting probability 1 on a single point.  Evaluating \cref{eq:TV-from-2} for such a distribution yields Aldous's inequality~\cite{aldous}, which bounds the mixing time by $\leq \delta^{-1}\log(1/\min_x \Pi(x))$.  Thus a warm start can be seen as avoiding the term $\log(1/\min_x \Pi(x))$, which often will be $O(n)$ for a Markov chain on $n$ bits.

The benefits of warm starts for quantum algorithms, specifically that of Wocjan-Abeyesinghe~\cite{wocjan}, are much higher.  Indeed, a reflection about $\ket\Pi$ takes time $O(1/\sqrt\delta)$, while mapping an arbitrary starting state $\ket\psi$ to $\ket\Pi$ using a generalized Grover algorithm takes $O(1/|\braket{\psi|\Pi}|)$ reflections.  Szegedy~\cite{szegedy} and MNRS~\cite{magniez} perform such a series of reflections to obtain a quantum walk search algorithm whose runtime scales as $O(1/\sqrt{\delta\min_x \Pi(x)})$, resulting in a dependence on overlap that is exponentially worse than the classical case in \cref{eq:TV-from-2}. By annealing through a judicious choice of starting states, Wocjan-Abeyesinghe~\cite{wocjan} avoid this term at the cost of introducing a dependence on $\ell$, the annealing schedule length.

\paragraph{Nonadaptive schedules.}
SVV focus specifically on the problem of approximate counting, not Bayesian inference, so they can use nonadaptive schedules to ensure warm starts at almost no additional cost. Suppose that we would like to construct an adaptive temperature schedule of length $\ell$.  In the case of approximate counting, where we need to estimate each of the $\ell$ terms in \cref{eq:telescope}, we need $O(\ell/\eps^2)$ (classical) samples at each temperature, incurring a total cost of $O(\ell^2/\eps^2)$.  (Note that this oversimplifies slightly and leaves out some additional factors.)   Since a nonadaptive schedule has length $O(\ell^2)$, taking one sample from each of the $O(\ell^2)$ temperatures would not lead to any asymptotic increase in cost.  For this reason SVV choose to begin with a nonadaptive schedule of length $O(\ell^2)$, where each temperature can be easily shown to provide a warm start for the next.  Then they can select a subset of $\ell$ temperatures to repeatedly sample in order to estimate the partition function.

Montanaro observed (see \cite[Section 3.3]{montanaro}) that this approach does not combine well with quantum walks.  Quantum walks cannot directly create states at a given temperature without prohibitive cost, and the no-cloning theorem means that we cannot keep copies of the states produced along the way without recreating them from scratch.  If we need one copy of each state at a sequence of $\ell$ temperatures then we need to run a quantum walk $(1 + 2 + \ldots + \ell)/\sqrt\delta = O(\ell^2/\sqrt\delta)$  times, which further increases to $O(\ell^3/\sqrt\delta)$ if we need to select $\ell$ temperatures out of a list of $\ell^2$ temperatures.   (We ignore the  dependence on accuracy and error probability here for simplicity.)

Our strategy for constructing the $\ell$-step adaptive schedule never needs to create an $O(\ell^2)$-step nonadaptive schedule, and this change did not require major new ideas.  However, it alone is not enough, because without the ability to reuse states we would still incur the $O(\ell^2/\sqrt{\delta})$ cost described above.

\paragraph{Non-destructive amplitude estimation.}
The missing ingredient in previous work is our \thmref{amplest1}, which shows that amplitude estimation can be made nondestructive.  We use this both to create the schedule and to estimate the terms $Z(\beta_{i+1})/Z(\beta_i)$ in \cref{eq:telescope}.  For Bayesian inference this is an important piece of our speedup, as it allows us to achieve time $\tilde O(\ell/\sqrt\delta)$ instead of $\tilde O(\ell^2/\sqrt\delta)$.  As a result it becomes worthwhile to drop the nonadaptive schedule of SVV. 
  For approximate counting we cannot avoid an $\ell^2$ dependence in our $\tilde O(\ell^2/\sqrt{\delta}\epsilon)$ runtime, but dropping the nonadaptive schedule does remove the additive term of $O(\ell^2/\delta)$ that appeared in \cite{montanaro}.

\subsection{Conclusion}\label{sec:conclusion}
To summarize, we have shown how to combine quantum simulated annealing with shorter adaptive annealing schedules, resulting in a QSA algorithm that displays a quadratic improvement in dependence on both schedule length and inverse spectral gap when compared against a nonadaptive classical annealing algorithm. We have demonstrated applications to Bayesian inference and estimating partition functions of counting problems, and in the process we have also shown that amplitude estimation can be made nondestructive, a result that is useful in its own right.

This paper can be viewed as part of the broader goal of finding quadratic (or other polynomial) speedups of as many general-purpose classical algorithms as possible.
Grover's algorithm can be interpreted as a square-root speedup for exhaustive search, and likewise there are easy quantum quadratic speedups for rejection sampling. However, the best classical algorithms for counting and Bayesian inference are much better than naive enumeration or rejection sampling. While simulated annealing with an adaptive schedule is still a generic algorithm, it is often much closer to the state of the art, and so it is worthwhile to try to find a quantum speedup for it.  We do not fully square root its runtime since our sequence length is essentially the same as the best classical result (instead of quadratically worse as in previous quantum results), but our runtime dependence on accuracy and spectral gap are both quadratically better than those of classical algorithms.

Within the paradigm of simulated annealing we are unlikely to see further improvements in sequence length or dependence on accuracy or spectral gap.  However, our algorithm for Bayesian inference does improve on classical algorithms by returning a qsample instead of a classical sample.  We hope that future algorithms will use this fact to find further quantum algorithmic advantages.

\appendix
\section{Bounding the Length of the Cooling Schedule}\label{app:schedlen}
Here we provide the proof of Lemma \ref{lemma:svv}, which is a slight modification of Lemma 4.3 in SVV~\cite{stefankovic}. We use this result to demonstrate the existence of a temperature schedule satisfying the bounded variance (\ref{eq:boundedvariance}) and slow-varying (\ref{eq:boundedslovar}) conditions, or equivalently (\ref{eq:approxconcav}), and to bound the length of such a schedule.
\begin{lemma}\label{lemma:svv2}
(Modified from SVV~\cite{stefankovic} Lemma 4.3) For $f$ a convex function over domain $[0, \gamma]$, there exists a sequence $\gamma_0<\gamma_1<\ldots<\gamma_\ell$ with $\gamma_0=0$ and $\gamma_\ell=\gamma$ satisfying
\begin{equation}\label{eq:approxconcav2}
f\left(\frac{\gamma_i+\gamma_{i+1}}{2}\right)\geq\frac{f(\gamma_i)+f(\gamma_{i+1})}{2}-1
\end{equation}
with length
\begin{equation}\label{eq:len2}
\ell\leq \sqrt{(f(0)-f(\gamma))\log\left(\frac{f'(0)}{f'(\gamma)+1}\right)}.
\end{equation}
\end{lemma}
\begin{proof}[Proof of Lemma \ref{lemma:svv2}]
Suppose we have already constructed the sequence up to $\gamma_i$. Let $\gamma_{i+1}$ be the largest value in $[\gamma_i, \gamma]$ so that $\gamma_i$ and $\gamma_{i+1}$ satisfy equation (\ref{eq:approxconcav2}), and let $m_i=(\gamma_i+\gamma_{i+1})/2$, $\Delta_i=(\gamma_{i+1}-\gamma_i)/2$, and $K_i=f(\gamma_i)-f(\gamma_{i+1})$. Then, since $\gamma_{i+1}$ satisfies equation (\ref{eq:approxconcav2}),
\begin{equation}\label{eq:midpt}
f(m_i)\geq\frac{f(\gamma_i)+f(\gamma_{i+1})}{2}-1.
\end{equation}
By convexity,
\begin{equation}
f'(\gamma_i)\leq \frac{f(\gamma_{i+1})-f(\gamma_i)}{\gamma_{i+1}-\gamma_i}.
\end{equation}
We can rewrite this as 
\begin{equation}\label{eq:gammai}
-f'(\gamma_i)\geq\frac{K_i}{2\Delta_i}.
\end{equation}
Also by convexity,
\begin{equation}
f'(\gamma_{i+1})\geq\frac{f(m_i)-f(\gamma_{i+1})}{m_i-\gamma_{i+1}}.
\end{equation}
Combining this with equation (\ref{eq:midpt}),
\begin{equation}\label{eq:gammai+1a}
-f'(\gamma_{i+1})\leq\frac{K_i-2}{2\Delta_i}
\end{equation}
and
\begin{equation}\label{eq:gammai+1b}
-f'(\gamma_{i+1})-1\leq\frac{K_i-2}{2\Delta_i}-1.
\end{equation}
Then, combining equations (\ref{eq:gammai}) and (\ref{eq:gammai+1a}),
\begin{equation}\label{eq:ratioa}
\frac{f'(\gamma_{i+1})}{f'(\gamma_i)}\leq 1-\frac{2}{K_i}\leq 1-\frac{1}{K_i}\leq e^{-1/K_i}.
\end{equation}
Similarly, combining equations (\ref{eq:gammai}) and (\ref{eq:gammai+1b}),
\begin{equation}\label{eq:ratiob}
\frac{f'(\gamma_{i+1})+1}{f'(\gamma_i)}\leq 1-\frac{2+2\Delta_i}{K_i}\leq 1-\frac{1}{K_i}\leq e^{-1/K_i}.
\end{equation}
Summing the $K_i$,
\begin{equation}\label{eq:kia}
\sum_{i=0}^{\ell-1}K_i= f(0)-f(\gamma).
\end{equation}
Summing equation (\ref{eq:ratioa}) over $K_i$ for $i=0$ to $\ell-2$, and adding equation (\ref{eq:ratiob}) for $i=\ell-1$, we get that
\begin{equation}\label{eq:kib}
\sum_{i=0}^{\ell-1}\frac{1}{K_i}\leq\log\left(\frac{f'(0)}{f'(\gamma)+1}\right).
\end{equation}
By the Cauchy-Schwarz inequality on equations (\ref{eq:kia}) and (\ref{eq:kib}),
\begin{equation}
\ell^2\leq(f(0)-f(\gamma))\log\left(\frac{f'(0)}{f'(\gamma)+1}\right).
\end{equation}
\end{proof}

\section*{Acknowledgements}
We would like to thank Ashley Montanaro and Fernando Brand\~ao for helpful conversations and suggestions. AYW would like to acknowledge the support of the DOE CSGF.  AWH was funded by NSF grants CCF-1452616, CCF-1729369, PHY-1818914, ARO contract W911NF-17-1-0433 and a Samsung Advanced Institute of Technology Global Research Partnership.


\begin{thebibliography}{10}

\bibitem{aharonov}
D.~Aharonov and A.~Ta-Shma.
\newblock {Adiabatic quantum state generation and statistical zero knowledge}.
\newblock In {\em Proceedings of the 35th Annual ACM Symposium on Theory of
  computing (STOC)}, pages 20--29. ACM Press New York, NY, USA, 2003,
  \href{http://arxiv.org/abs/quant-ph/0301023}{{\ttfamily
  arXiv:quant-ph/0301023}}.

\bibitem{aldous}
D.~Aldous.
\newblock Some inequalities for reversible {Ma}rkov chains.
\newblock {\em Journal of the London Mathematical Society}, 25:564--576, 1982.

\bibitem{amb}
A.~Ambainis, A.~Gilyen, S.~Jeffery, and M.~Kokainis.
\newblock Quantum speedup for finding marked vertices by quantum walks, 2019,
  \href{http://arxiv.org/abs/1903.07493}{{\ttfamily arXiv:1903.07493}}.

\bibitem{as}
S.~Apers and A.~Sarlette.
\newblock Quantum fast-forwarding: {M}arkov chains and graph property testing,
  2018,  \href{http://arxiv.org/abs/1804.02321}{{\ttfamily arXiv:1804.02321}}.

\bibitem{brassard}
G.~Brassard, P.~H{\o}yer, M.~Mosca, and A.~Tapp.
\newblock {\em Quantum Amplitude Amplification and Estimation}, volume 305 of
  {\em Contemporary Mathematics Series Millenium Volume}.
\newblock AMS, 2002,  \href{http://arxiv.org/abs/quant-ph/0005055}{{\ttfamily
  arXiv:quant-ph/0005055}}.

\bibitem{dfk}
M.~Dyer, A.~Frieze, and R.~Kanna.
\newblock A random polynomial time algorithm for approximating the volume of
  convex bodies.
\newblock {\em Journal of the ACM}, 38(1):1--17, 1991.

\bibitem{groverrudolph}
L.~Grover and T.~Rudolph.
\newblock Creating superpositions that correspond to efficiently integrable
  probability distributions, 2002,
  \href{http://arxiv.org/abs/quant-ph/0208112}{{\ttfamily
  arXiv:quant-ph/0208112}}.

\bibitem{huber}
M.~{Huber}.
\newblock {Approximation algorithms for the normalizing constant of Gibbs
  distributions}.
\newblock {\em arXiv e-prints}, page arXiv:1206.2689, Jun 2012,
  \href{http://arxiv.org/abs/1206.2689}{{\ttfamily arXiv:1206.2689}}.

\bibitem{jerrum1}
M.~Jerrum.
\newblock A very simple algorithm for estimating the number of {$k$}-colorings
  of a low-degree graph.
\newblock {\em Random Structures \& Algorithms}, 7(2):157--165, 1995.

\bibitem{js}
M.~Jerrum and A.~Sinclair.
\newblock Approximating the permanent.
\newblock {\em SIAM Journal on Computing}, 18:1149--1178, 1989.

\bibitem{jsv}
M.~Jerrum, A.~Sinclair, and E.~Vigoda.
\newblock A polynomial-time approximation algorithm for the permanent of a
  matrix with nonnegative entries.
\newblock {\em J. ACM}, 51(4):671--697, 2004.

\bibitem{powering}
M.~Jerrum, L.~Valiant, and V.~Vazirani.
\newblock Random generation of combinatorial structures from a uniform
  distribution.
\newblock {\em Theoretical computer science}, 43(2-3):169--188, 1986.

\bibitem{kayemosca}
P.~{Kaye} and M.~{Mosca}.
\newblock {Quantum Networks for Generating Arbitrary Quantum States}.
\newblock {\em arXiv e-prints}, pages quant--ph/0407102, Jul 2004,
  \href{http://arxiv.org/abs/quant-ph/0407102}{{\ttfamily
  arXiv:quant-ph/0407102}}.

\bibitem{low}
G.~H. Low, T.~J. Yoder, and I.~L. Chuang.
\newblock Quantum inference on bayesian networks.
\newblock {\em Physical Review A}, 89(6):062315, 2014,
  \href{http://arxiv.org/abs/1402.7359}{{\ttfamily arXiv:1402.7359}}.

\bibitem{magniez}
F.~Magniez, A.~Nayak, J.~Roland, and M.~Santha.
\newblock Search via quantum walk.
\newblock {\em SIAM Journal on Computing}, 40(1):142--164, 2011,
  \href{http://arxiv.org/abs/quant-ph/0608026}{{\ttfamily
  arXiv:quant-ph/0608026}}.

\bibitem{qma}
C.~Marriott and J.~Watrous.
\newblock Quantum arthur-merlin games.
\newblock {\em Computational Complexity}, 14(2):122--152, 2005,
  \href{http://arxiv.org/abs/cs/0506068}{{\ttfamily arXiv:cs/0506068}}.

\bibitem{martinelli}
F.~Martinelli and E.~Olivieri.
\newblock Approach to equilibrium of {G}lauber dynamics in the one phase
  region.
\newblock {\em Communications in Mathematical Physics}, 161(3):447--486, 1994.

\bibitem{montanaro}
A.~Montanaro.
\newblock Quantum speedup of {M}onte {C}arlo methods.
\newblock {\em Proceedings of the Royal Society of London A: Mathematical,
  Physical and Engineering Sciences}, 471(2181), 2015,
  \href{http://arxiv.org/abs/1504.06987}{{\ttfamily arXiv:1504.06987}}.

\bibitem{mossel}
E.~Mossel, A.~Sly, et~al.
\newblock Exact thresholds for {Ising--Gibbs} samplers on general graphs.
\newblock {\em The Annals of Probability}, 41(1):294--328, 2013.

\bibitem{orsucci}
D.~Orsucci, H.~J. Briegel, V.~Dunjko, et~al.
\newblock Faster quantum mixing for slowly evolving sequences of {M}arkov
  chains.
\newblock {\em Quantum}, 2:105, 2018,
  \href{http://arxiv.org/abs/1503.01334}{{\ttfamily arXiv:1503.01334}}.

\bibitem{rejectionsampling}
M.~Ozols, M.~Roetteler, and J.~Roland.
\newblock Quantum rejection sampling.
\newblock {\em ACM Transactions on Computation Theory (TOCT)},
  5(3):11:1--11:33, 2013,  \href{http://arxiv.org/abs/1103.2774}{{\ttfamily
  arXiv:1103.2774}}.

\bibitem{richter}
P.~C. Richter.
\newblock Quantum speedup of classical mixing processes.
\newblock {\em Physical Review A}, 76(4):042306, 2007,
  \href{http://arxiv.org/abs/quant-ph/0609204}{{\ttfamily
  arXiv:quant-ph/0609204}}.

\bibitem{somma2}
R.~Somma, S.~Boixo, and H.~Barnum.
\newblock Quantum simulated annealing, 2007,
  \href{http://arxiv.org/abs/0712.1008}{{\ttfamily arXiv:0712.1008}}.

\bibitem{somma}
R.~Somma, S.~Boixo, H.~Barnum, and E.~Knill.
\newblock Quantum simulations of classical annealing processes.
\newblock {\em Phys.~Rev.~Lett.}, 101(13):130504, 2008,
  \href{http://arxiv.org/abs/0804.1571}{{\ttfamily arXiv:0804.1571}}.

\bibitem{stefankovic}
D.~{\v{S}}tefankovi{\v{c}}, S.~Vempala, and E.~Vigoda.
\newblock Adaptive simulated annealing: A near-optimal connection between
  sampling and counting.
\newblock {\em Journal of the ACM (JACM)}, 56(3):18, 2009,
  \href{http://arxiv.org/abs/cs.DS/0612058}{{\ttfamily arXiv:cs.DS/0612058}}.

\bibitem{szegedy}
M.~Szegedy.
\newblock Quantum speed-up of {Markov} chain based algorithms.
\newblock In {\em FOCS '04: Proceedings of the 45th Annual IEEE Symposium on
  Foundations of Computer Science}, pages 32--41, Washington, DC, USA, 2004.
  IEEE Computer Society,
  \href{http://arxiv.org/abs/quant-ph/0401053}{{\ttfamily
  arXiv:quant-ph/0401053}}.

\bibitem{Tang19}
E.~Tang.
\newblock Some settings supporting efficient state preparation.
\newblock
  \url{https://ewintang.com/blog/2019/06/13/some-settings-supporting-efficient-state-preparation/},
  2019.

\bibitem{temme}
K.~Temme, T.~J. Osborne, K.~G. Vollbrecht, D.~Poulin, and F.~Verstraete.
\newblock Quantum {M}etropolis sampling.
\newblock {\em Nature}, 471(7336):87--90, 2011,
  \href{http://arxiv.org/abs/0911.3635}{{\ttfamily arXiv:0911.3635}}.

\bibitem{vigoda}
E.~Vigoda.
\newblock A note on the glauber dynamics for sampling independent sets.
\newblock {\em The Electronic Journal of Combinatorics}, 8(1):8, 2001.

\bibitem{wiebe}
N.~Wiebe and C.~Granade.
\newblock Can small quantum systems learn?, 2015,
  \href{http://arxiv.org/abs/1512.03145}{{\ttfamily arXiv:1512.03145}}.

\bibitem{wocjan}
P.~Wocjan and A.~Abeyesinghe.
\newblock Speedup via quantum sampling.
\newblock {\em Phys. Rev. A}, 78:042336, 2008,
  \href{http://arxiv.org/abs/0804.4259}{{\ttfamily arXiv:0804.4259}}.

\bibitem{yungaag}
M.-H. Yung and A.~Aspuru-Guzik.
\newblock A quantum--quantum metropolis algorithm.
\newblock {\em Proceedings of the National Academy of Sciences},
  109(3):754--759, 2012,  \href{http://arxiv.org/abs/1011.1468}{{\ttfamily
  arXiv:1011.1468}}.

\bibitem{zalka}
C.~Zalka.
\newblock {Efficient simulation of quantum systems by quantum computers}.
\newblock {\em Proc. Roy. Soc. Lond.}, A454:313--322, 1998,
  \href{http://arxiv.org/abs/quant-ph/9603026}{{\ttfamily
  arXiv:quant-ph/9603026}}.

\end{thebibliography}

\end{document}